\documentclass[journal,twoside,letterpaper,final]{IEEEtran}
\IEEEoverridecommandlockouts
\usepackage{enumitem}
\usepackage{amsmath, amsthm}
\usepackage{bbm}
\usepackage{graphicx}
\usepackage[usenames,dvipsnames,svgnames,table]{xcolor}
\usepackage{booktabs}
\usepackage{array}
\usepackage{amssymb}
\usepackage{subfigure}
\usepackage{setspace}
\usepackage{epstopdf}
\usepackage[normalem]{ulem}
\usepackage{cite}
\usepackage[skip=0pt]{caption}
\allowdisplaybreaks
\def\com#1{\textcolor{Black}{#1}}
\def\comone#1{\textcolor{Black}{#1}}
\def\comtwo#1{\textcolor{Black}{#1}}
\def\comthree#1{\textcolor{Black}{#1}}
\def\comnew#1{\textcolor{Black}{#1}}

\newtheorem{theorem}{Theorem}

\newtheorem{proposition}[theorem]{Proposition}
\newtheorem{corollary}[theorem]{Corollary}

\theoremstyle{definition}
\newtheorem{defn}{Definition}

\theoremstyle{remark}

\theoremstyle{remark}

\newlist{Properties}{enumerate}{2}
\setlist[Properties]{label=\textit{Property} \textit{\arabic*.},itemindent=*}

\begin{document}
\title{Directional and Causal Information Flow in EEG for Assessing Perceived Audio Quality}

\author{\IEEEauthorblockN{Ketan Mehta\IEEEauthorrefmark{1} and J\"{o}rg Kliewer\IEEEauthorrefmark{2}, \IEEEmembership{Senior Member,~IEEE}\thanks{This work was supported in part by NSF grant CCF-1065603, and presented in part at the 49th Asilomar Conference on Signals, Systems and Computers \cite{MehtaAsilomar2015}, and at the 2017 IEEE International Conference on Communications \cite{mehta2017icc}.}}\\
\IEEEauthorblockA{\IEEEauthorrefmark{1}Klipsch School of Electrical and Computer Engineering, \\ New Mexico State University, NM 88003}
\\
\IEEEauthorblockA{\IEEEauthorrefmark{2}Helen and John C.~Hartmann Dept. of Electrical and Computer Engineering,\\ New Jersey Institute of Technology, NJ 07103}
}%
\maketitle

\begin{abstract}
In this paper, electroencephalography (EEG) measurements are used to infer change in cortical functional connectivity in response to change in audio stimulus. Experiments are conducted wherein the EEG activity of human subjects is recorded as they listen to audio sequences whose quality varies with time. A causal information theoretic framework is then proposed to measure the information flow between EEG sensors appropriately grouped into different regions of interest (ROI) over the cortex. A new causal bidirectional information (CBI) measure is defined as an improvement over standard directed information measures for the purposes of identifying connectivity between ROIs in a generalized cortical network setting. CBI can be intuitively interpreted as a causal bidirectional modification of directed information, and inherently calculates the divergence of the observed data from a multiple access channel with feedback. Further, we determine the analytical relationship between the different causal measures and compare how well they are able to distinguish between the perceived audio quality. The connectivity results inferred indicate a significant change in the rate of information flow between ROIs as the subjects listen to different audio qualities, with CBI being the best in discriminating between the perceived audio quality, compared to using standard directed information measures.
\end{abstract}


\vspace{-2ex}
\section{Introduction}
\label{sec:intro}

\vspace{-1ex}

%

Detection and response to stimuli is in general a multistage process that results in the hierarchical activation and interaction of several different regions in the brain. To understand the dynamics of brain functioning it therefore essential to investigate the information flow and connectivity (interactions) between different regions in the brain. Further, in each of these hierarchies, sensory and motor information in the brain is represented and manipulated in the form of neural activity patterns. The superposition of this electrophysiological activity can be recorded via electrodes on the scalp and is termed as electroencephalography (EEG).

Functional connectivity refers to the statistical dependencies between neural data recorded from spatially distinct regions in the brain \cite{friston1994functional, horwitz2003elusive}. 
Information theory provides a stochastic framework which is fundamentally well suited for the task of assessing functional connectivity \cite{dimitrov2011information, spikes} between neural responses. For example, \cite{Paninski, Panzeri} presented \comtwo{mutual information (MI) \cite{cover2012elements}} estimates to assess the correlation between the spike timing of an ensemble of neurons. Likewise, \cite{ohiorhenuan2010sparse, schneidman2006weak, shlens2009structure} investigated the effectiveness of calculating pairwise maximum entropy to model the activity of a larger population of neurons.
MI has also has been successfully employed in the past for determining the functional connectivity in EEG sensors for feature extraction and classification purposes \cite{Xu, wu2007classifying}.
Similarly, other studies have used MI to analyze EEG data to investigate corticocortical information transmission for pathological conditions such as Alzheimer's disease \cite{Alzh} and schizophrenia
\cite{Schiz}, or for odor stimulation \cite{Odor1,Odor2}.

One limitation of MI and entropy when applied in the traditional (Shannon) sense is their inability to distinguish between the direction of information flow, as pointed out by Marko in \cite{marko1973bidirectional}. In the same work Marko also proposed to calculate the information flow in each direction of a bidirectional channel using conditional probabilities based on Markovian dependencies. In \cite{massey1990causality}, Massey extended the initial work by Marko and formally defined directed information as the information flow from the input to the output of a channel with feedback. Other measures have similarly been defined for calculating the directional information transfer rate between random processes, most notably, Kamitake's directed information \cite{Kamitake} and transfer entropy by Schreiber \cite{schreiber2000measuring}. Further, feedback and directionality are also closely related to the notion of \textit{causality} in information measures\footnote{Causality and directionality are formally defined in Def.~1 and Def.~2, resp., in Sec. III.} \cite{amblard2011directed, kramerthesis, massey1990causality}. Massey's directed information and transfer entropy are in general referred to as \textit{causal}, since they measure statistical dependencies between the \textit{past} and current values of a process. We adopt this definition in this paper and therefore, causality here takes on the usual meaning of a cause occurring prior to its effect, or a stimulus occurring before response, i.e., how the past states of a system influences its present and future states \cite{quinn2015directed}.
Our interest here is in using EEG for assessing human perception of time-varying audio quality. We are inspired by our recent results in \cite{mehtaMBMC} which uses MI to quantify the information flow over the end-to-end perceptual processing chain from audio stimulus to EEG output. One characteristic common in subjective audio testing protocols including the current state-of-the-art approach, Multi Stimulus with Hidden Anchor (MUSHRA) \cite{MUSHRA}, is that they require human participants to assign a single quality-rating score to each test sequence. \comtwo{Such conventional testing suffers from a subject-based bias towards cultural factors in the local testing environment and can tend to be highly variable. Instead, neurophysiological measurements such as EEG directly capture and analyze the brainwave response patterns that depend only on the perceived variation in signal quality \cite{bosse2016brain, engelke2017psychophysiology}. As a result, EEG is inherently well suited to assess human perception of audio \cite{creusere2012assessment, porbadnigk2010using} and visual \cite{porbadnigk2010using, scholler2012toward, mustafa2012single} quality. For example, in \cite{porbadnigk2010using, scholler2012toward} the authors used linear discriminant analysis classifiers to extract features from EEG for classifying noise detection in audio signals and to assess changes in perceptual video quality, respectively. Similarly, \cite{creusere2012assessment} identified features in EEG brainwave responses corresponding to time-varying audio quality using a time-space-frequency analysis, while \cite{mustafa2012single} employed a wavelet-based approach for an EEG classification of commonly occurring artifacts in compressed video, using a single-trial EEG.}




To the best of our knowledge, however, the work presented here is the first time that functional connectivity has been applied in conjunction with EEG measurements for the purposes of assessing audio quality perception. By using causal information measures to \comtwo{detect a change in functional connectivity} we directly identify those cortical regions which are most actively involved in perceiving a \comtwo{change in audio quality}. Further, we establish the analytical relationship between the different presented information measures and compare how well each of them is able to distinguish between the perceived audio qualities. Towards this end, we consider two distinct scenarios for estimating the connectivity between EEG sensors by appropriately grouping them into regions of interest (ROIs) over the cortex. In the first scenario, we employ Massey's and Kamitake's directed information and transfer entropy, respectively, to calculate the pairwise directional information flow between ROIs while using causal conditioning to account for the influence from all other regions. In the second scenario we propose a novel information measure which can be considered as a causal bidirectional modification of directed information applied to a generalized cortical network setting. In particular, we show that the proposed causal bidirectional information (CBI) measure assesses the direct connectivity between any two given nodes of a multiterminal cortical network by inherently calculating the divergence of the induced conditional distributions from those associated with a multiple access channel (MAC) with feedback. Each presented measure is validated by applying it to analyze real EEG data recorded for human subjects as they listen to audio sequences whose quality changes over time. For the sake of simplicity and analytical tractability we restrict ourselves to only two levels of audio quality (high quality and degraded quality). We determine and compare the instantaneous information transfer rates as inferred by each of these measures for the case where the subject listens to high quality audio as opposed to the case when the subject listens to degraded quality audio. Finally, note that we are not able make \comone{any assumptions about the actual structure of the underlying cortical channels (e.g., linear vs. non-linear) as our analysis is solely based on the observed empirical distributions of the data at the input and output of these channels.} 

The rest of the paper is organized as follows. Section II provides an overview of EEG, the experiment, and the stimulus audio sequences. In Section III we review some directional information measures widely used in the literature for estimating connectivity. We assess the information flow between cortical regions using directional information measures in Section IV, along with determining the analytical relationship between these measures. In Section V we introduce CBI and discuss its properties. The results of our analysis on EEG data are presented in Section VI. We finally conclude with a summary of our study and future directions in Section VII.

\vspace{-2ex}
\section{Background}
\label{sec:methods}

In the conducted study, the EEG response activity of human test subjects is recorded as they listen to a variety of audio test-sequences. The quality of these stimulus test-sequences is varied with time between different ``quality levels". 
All audio test-sequences were created from three fundamentally different base-sequences sampled at a reference base quality of 44.1\,kHz, with a precision of 16 bits per sample. Here, we employ the same test-sequences and distortion quality levels as in \cite{creusere2011assessing}.\
Two different types of distortions are considered for our analysis, scalar quantization and frequency band truncation, \comtwo{where the specific parameters are listed in Table 1.} The test-sequence for a specific trial is created by selecting one of the two distortion types and then applying it to the original base-sequence in a time-varying pattern of non-overlapping five second blocks as shown in Fig.~\ref{fig:aq}. Multiple of such trials are conducted for each subject by choosing all possible combinations of sequences, distortion types, and time-varying patterns. Note that despite the subjects being presented with all different quality levels in our listening tests, here we focus exemplary only on the ``high" base-quality and the ``Q3 degraded" quality audio. This addresses the worst-case quality change and keeps the problem analytically and numerically tractable. A detailed exposition of the experimental setup, test-sequences, and distortion quality levels is provided in \cite{mehtaMBMC}.

\renewcommand{\arraystretch}{1.2}
\begin{table}[bp]
	\vspace{-1ex}
	\centering
	\small
	\caption{\small{\com{Different quality levels presented to the subject during
		the course of an audio trial. To generate the distortion, each of
		these base-sequences were passed through a 2048-point modified discrete cosine transform and either frequency truncation or scalar quantization was applied to the coefficients prior to reconstruction \cite{creusere2011assessing}.}}}
	\begin{tabular}{@{}r|c|c@{}}
		\toprule
		Quality Level & Freq. Truncation & Scalar Quantization \\
		& Low Pass Filter & No.~of Significant Bits Retained \\
		\midrule
		Q1 	& 	4.4 KHz 	&	4	 \\
		Q2 	& 	2.2 KHz 	&	3	 \\
		Q3 	& 	1.1 KHz		&	2	 \\ 
		\bottomrule
	\end{tabular}
	\label{tab:audioquality}
\vspace{-1ex}
\end{table}
\begin{figure}[htbp]
	\vspace{-1ex}
	\centering
		\includegraphics[width = \columnwidth]{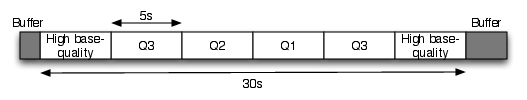}
	\caption{\small{A 30 second test sequence where the audio quality changes in a time-varying pattern over the whole duration of the sequence. Different possible combinations of quality changes and two distortion types are presented to each subject in a randomized fashion (adopted from \protect\cite{mehtaMBMC}).}}
	\label{fig:aq}
	\vspace{-2ex}
\end{figure}



%


The EEG data is captured on a total of 128 spatial channels using an ActiveTwo Biosemi system with a sampling rate of 256\,Hz. To better manage the large amount of collected data while also effectively covering the activity over different regions of the cortex, we group the 128 electrodes into specific regions of interest (ROI) as shown in Fig.~\ref{fig:ROI}. While a large number of potential grouping schemes are possible, this scheme is favored for our purposes as it efficiently covers all the cortical regions (lobes) of the brain with a relatively low number of ROIs. Also, the number of electrodes in any given ROI varies between a minimum of 9 to a maximum of 12. For example, in our region partitioning scheme ROI~2 (9 electrodes) covers the prefrontal cortex, ROI~6 (10 electrodes) the parietal lobe, ROI~8 (9 electrodes) the occipital lobe, and ROI~5 and ROI~7 (12 electrodes each) cover the left and right temporal lobes, respectively. In essence, our goal here then is to investigate the causal connectivity between the different ROI in response to the different audio quality levels.

\begin{figure}[btp]
	\vspace{-1ex}
	\centering
	\includegraphics[scale=0.22]{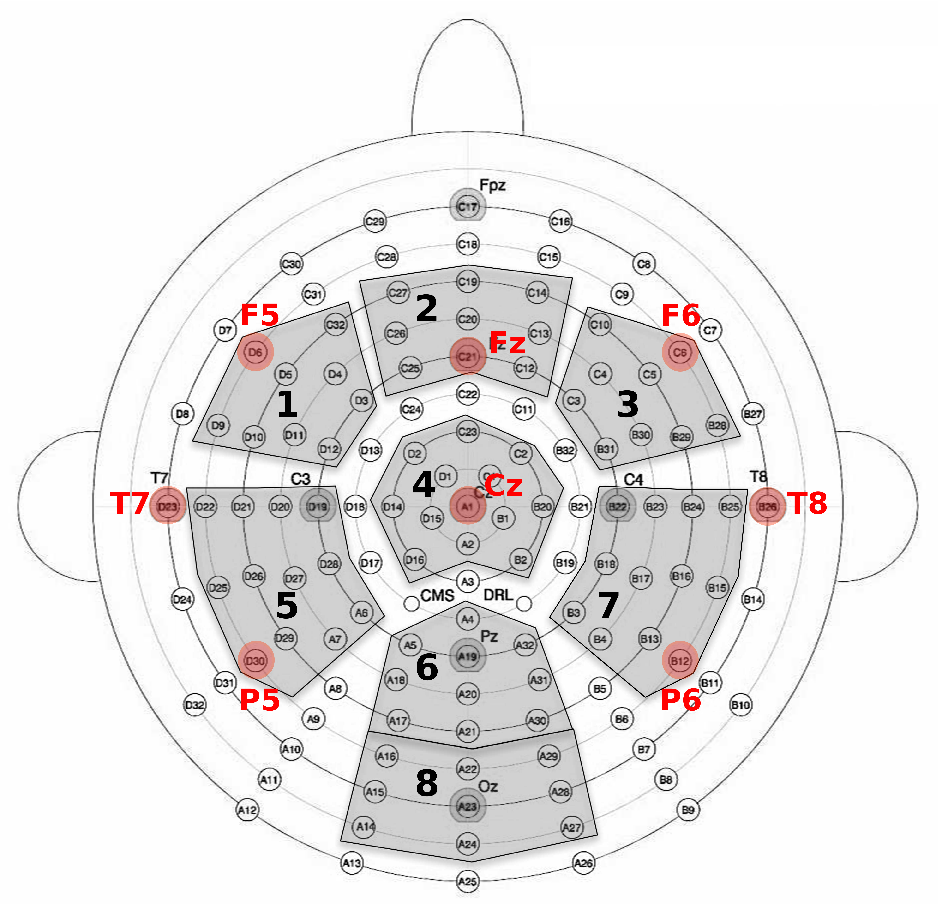}
	\caption{\small{The 128 electrodes are grouped into eight regions of
			interest (ROI) to effectively cover the different cortical regions
			(lobes) of the brain (adopted from \protect\cite{mehtaMBMC}). The red electrode locations are used for generating synthetic EEG data in Sec.~VI-C, see Table~\ref{tab:EEGsimulated} for the naming convention.}}
	\label{fig:ROI}
	\vspace{-4ex}
\end{figure}

\vspace{-2ex}
\section{Directionality, Causality and Feedback in Information Measures}
\vspace{-1ex}
Let $X^n$ denote a random vector of $n$ constituent discrete valued random variables $[X_1, X_2, \ldots, X_n]$. Also, let $x^n = [x_1, x_2, \ldots, x_n]$ be the corresponding realizations drawn from the joint probability distribution denoted by $p(x_1, x_2, \ldots , x_n)$, respectively. Similarly, $X_n^N = [X_n , X_{n+1} , \cdots, X_N]$ represents a length $N-n+1$ random vector. 
Denoting the expected value of a random variable by $\mathbb{E}[\cdot]$, the entropy of the $n$-tuple $X^n$ can be written as $H(X^n) = -\mathbb{E}[\log p(X^n)].$

The mutual information (MI) between two length $N$ interacting random processes $X^N$ and $Y^N$ is defined as 
\begin{align}
I(X^N;Y^N) 	&= \sum\limits_{n=1}^{N} I(X^N;Y_n|Y^{n-1}) \label{eq:chain} \\
			&= H(Y^N) - H(Y^N|X^N),
\end{align}
with the conditional entropy $H(Y^n|X^n)\!=\!-\mathbb{E}[\log p(Y^n|X^n)]$ and
\begin{align}
H(Y^N|X^N) =  \sum\limits_{n=1}^{N} H(Y_n|Y^{n-1}X^{N}). \label{eq:conditionalentropy}
\end{align}
The MI measures the reduction in the uncertainty of $Y^N$ due to the knowledge of $X^N$, and is zero if and only if the two processes are statistically independent.
The MI between the two random processes is symmetric, i.e., $I(X^N;Y^N)=I(Y^N;X^N)$, and can therefore not distinguish between the direction of the information flow. Alternatively, a directional information measure introduces the notion of direction in the exchange of information between sources. In this paper we define a directional information measure as follows.


\begin{defn}
\textit{A directional information measure from $X^N$ to $Y^N$, represented with an arrow $X^N \rightarrow Y^N$, quantifies the information exchange rate in the direction from the input process $X^N$ towards the output process $Y^N$.}
\end{defn}

As implied by the definition, the information flow measured by a directional information measure is not symmetric, and in general the flow from $X^N \rightarrow Y^N$ is not equal to $Y^N\rightarrow X^N$.  In the following, we will examine three different directional information measures presented in the literature.

\vspace{-2ex}
\subsection{Massey's directed information}
 \vspace{-1ex}
Directed information as proposed by Massey in \cite{massey1990causality} is an extension of the preliminary work by Marko \cite{marko1973bidirectional} to characterize the information flow on a communication channel with feedback. Given input $X^N$ and output $Y^N$, the channel is said to be used without feedback if 
\begin{align}
p(x_n|y^{n-1}x^{n-1})=p(x_n|x^{n-1}), \quad\, \forall \,\,\, n \leq N, \label{eq:feedback}
\end{align}
i.e., the current channel input value does not depend on the past output samples. If the channel is used with feedback then \cite{massey1990causality} shows that the chain rule of conditional probability can be reduced to
\begin{align}
p(y^N|x^N) = \prod\limits_{n=1}^N p(y_n|y^{n-1}x^n). \label{eq:feedback2}
\end{align}
In closely related work \cite{Kramer98} introduced the concept of causal conditioning based on \eqref{eq:feedback2}. The entropy of $Y^N$ \textit{causally} conditioned on $X^N$ is defined as
\begin{align}
H(Y^N||X^N) = \sum\limits_{n=1}^{N} H(Y_n|Y^{n-1}X^n).
\end{align}
\begin{defn}
\textit{A measure between two random processes $X^N$ and $Y^N$ is said to be causal if the information transfer rate at time $n$ relies only on the dependencies between their past and current sample values $X^n$ and $Y^n$, and is not a function of the future sample values $X^N_{n+1}$ and $Y^N_{n+1}$.}
\end{defn}
\comnew{Therefore, the notion of causality as used in this work is based on inferring the statistical dependencies between the past states of a system on its present and future states \cite{amblard2011directed, kramerthesis, massey1990causality, quinn2015directed}. This is in contrast to the stronger interventional interpretation about casual inferences such as in \cite{pearl2003causality} which draws conclusions about \textit{causation}, e.g., process $X^N$ causes $Y^N$.}

Massey's directed information is a causal measure between sequence $X^N$ and $Y^N$ defined as
\begin{align} 
DI_1(X^N\!\rightarrow\! Y^N) &\triangleq H(Y^N)-H(Y^N||X^N) \label{eq:DI1a}\\
				&= \sum\limits_{n=1}^{N} I(X^n;Y_n|Y^{n-1}) \label{eq:DI1b}.
\end{align}
Equivalently, the directed information can also be written in terms of the Kullback-Leibler (KL) divergence as
\begin{align}
DI_1(X^N\!\rightarrow\! Y^N) &= \sum\limits_{n=1}^{N} D_{KL}\Big( p(y_n|y^{n-1}x^{n})\,||\,p(y_n|y^{n-1}) \Big) \label{eq:masseyKL}\\
&= \sum\limits_{n=1}^{N} \mathbb{E}\left[ \log \frac{p(Y_n|Y^{n-1}X^{n})}{p(Y_n|Y^{n-1})} \right].
\end{align}
Also in general, 
\begin{align}
DI_1(X^N\!\rightarrow\! Y^N) \leq I(X^N;Y^N), 
\end{align}
with equality if the channel is used without feedback. Directed information therefore not only gives a meaningful notation to the directivity of information, but also provides a tighter characterization than MI on the total information flow over a channel with feedback.
 \vspace{-1ex}
\subsection{Transfer entropy} 
\vspace{-1ex}
In \cite{schreiber2000measuring}, Schreiber introduced a causal measure for the directed exchange of information between two random processes called transfer entropy. 
Similar to the pioneering work in \cite{marko1973bidirectional}, transfer entropy considers a bidirectional communication channel between $X^N$ and $Y^N$, and measures the deviation of the observed distribution $p(y_n|y^{n-1}x^{n-1})$ over this channel from the Markov assumption
\begin{align}
p(y_n|y^{n-1}x^{n-1})=p(y_n|y^{n-1}), \qquad n \leq N. \label{eq:TEfeedback}
\end{align}
In particular, transfer entropy quantifies the deviation of the l.h.s. of \eqref{eq:TEfeedback} from the r.h.s. of \eqref{eq:TEfeedback} using the KL divergence and is defined as
\begin{align}  
TE(X^{n-1} \!\rightarrow\! Y^n ) &\triangleq D_{KL}\Big( p(y_n|y^{n-1}x^{n-1})\,||\,p(y_n|y^{n-1}) \Big) \label{eq:TE1}\\
&= \mathbb{E}\left[ \log \frac{p(Y_n|Y^{n-1}X^{n-1})}{p(Y_n|Y^{n-1})} \right] \label{eq:TE2} \\
&= H(Y_n|Y^{n-1})- H(Y_n|Y^{n-1}X^{n-1}) \label{eq:TE3}\\
&= I(Y_n;X^{n-1}|Y^{n-1}) \label{eq:TE4}.
\end{align}

\comnew{Correspondingly, for random processes of block length $N$ we can then define a sum transfer entropy \cite{ito2016backward, wibral2014directed} which in effect calculates and adds the transfer entropy at every history depth $n$,
\begin{align}
{TE}^\ast(X^{N-1} \!\rightarrow\! Y^N ) = \sum\limits_{n=1}^{N} TE(X^{n-1} \rightarrow Y^n).
\end{align}
}

Another widely used and closely related measure for causal influence was developed by Granger \cite{granger1969investigating}. Granger causality is a directional measure of statistical dependency based on prediction via vector auto-regression. The relationship between Granger causality and directional information measures has been previously analyzed in \cite{barnett2009granger, quinn2011estimating}.
For the specific case of Gaussian random variables, as it is the case in our EEG scenario, transfer entropy has been shown to be equivalent to Granger causality.

\vspace{-2ex}
\subsection{Kamitake's directed information}
Another variant of directed information is defined by Kamitake in \cite{Kamitake} and
given as 
\begin{align} \label{eq:DI2} 
DI_2(X^N\rightarrow Y^N) &\triangleq \sum\limits_{n=1}^{N} I(X_n;Y^{N}_{n+1}|X^{n-1}Y^{n}) \\
&\!\!=\! H(Y^{N}_{n+1}|X^{n-1}Y^n) \!-\! H(Y^{N}_{n+1}|X^{n}Y^n).
\end{align}
We notice that this measure is different from Massey's directed information in that it measures the influence of the current sample $X_n$ of $X^N$ at time $n$ on the
future samples $Y^{N}_{n+1}$ of $Y^N$. Kamitake's information measure is therefore directional, but not causal.


\section{Assessing Cortical Information Flow via Directional Measures}

\subsection{Causal conditioning and indirect influences in multiterminal networks}

A multi-terminal network characterizes the information flow between multiple communicating nodes with several senders and receivers. Let us denote three communicating nodes $\mathcal{X}$, $\mathcal{Y}$, and $\mathcal{Z}$ and the corresponding random processes associated with them as $X^N$, $Y^N$, and $Z^N$ respectively. \com{In our case, the information transfer over the cortex can be considered equivalent to a cortical multi-terminal network, with each ROI taking over the role of a communicating node. In the context of the cortical network the quantities $X^N$, $Y^N$, and $Z^N$, then correspond to the sampled EEG signals from different ROIs. Also without any loss of generality, $Z^N$ represents the output of multiple (and potentially all other) ROIs.
Our goal here is to identify the connectivity between the processes in the cortical network.  
In particular, there are two distinct instances of connectivity that can arise as a result of using directional informational measures.}

\begin{defn}\label{defn:directconnect}
\textit{A direct connectivity is said to exist from node $\mathcal{X}$ to node $\mathcal{Y}$ if there exists a non-zero information flow via a direct path between the nodes.}
\end{defn}
\begin{defn}
\textit{An implied connectivity arises when there is no direct path between two nodes $\mathcal{X}$ and $\mathcal{Y}$, but there is a non-zero information flow between the nodes because of an influence through other nodes in the network.}
\end{defn}
\begin{figure}[tbp]
\vspace{-1ex}
\centering
\hspace{5em}
\includegraphics[width = \columnwidth]{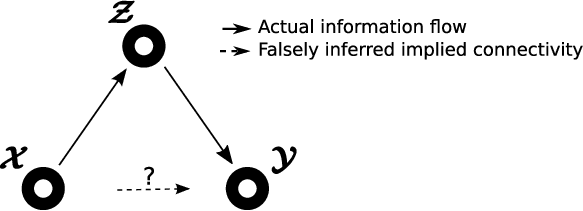}
\vspace{1ex}
\caption{\small{The relay channel is an example of a network in which simply using directed information can lead to a false inference of \textit{direct} connectivity between the nodes. In the example shown, $DI(X^N\!\rightarrow\!Y^N)>0$, even though there is no link between $X$ and $Y$.}}  \label{fig:broadrelay}
\vspace{-2ex}
\end{figure}
Therefore, a positive directed information between the random processes associated with any two nodes in a multi-node network alone does not necessarily equate to a direct connectivity between them \cite{quinn2011estimating, soltani2013inferring, amblard2014causal}. In Fig.~\ref{fig:broadrelay} we show an example network topology to illustrate how implied connectivity can lead to false inferences. In the shown relay channel there is no direct information transfer between $\mathcal{X}$ and $\mathcal{Y}$, but  the information flow is from $\mathcal{X}$ to $\mathcal{Z}$ to $\mathcal{Y}$. There is a Markovian influence $\mathcal{X} \!\rightarrow\! \mathcal{Z} \!\rightarrow\! \mathcal{Y}$. This results in a positive value for Massey's directed information, $DI_1(X^N\!\rightarrow\! Y^N) > 0$, thereby leading to implied connectivity between $\mathcal{X}$ and $\mathcal{Y}$.

Notice however that in the example presented, the knowledge of $Z^N$ leads to statistical conditional independence between $X^N$ and $Y^N$. We expand upon this idea and extend the expression for Massey's directed information to account for the influence of the additional random processes in the network via causal conditioning \cite{Kramer98}. 

\begin{defn} \label{defn:Masseycausal}
\textit{Causally conditioned Massey's directed information is defined as the information flowing from $X^N$ to $Y^N$ causally conditioned on the sequence $Z^{N-1}$ as}
\begin{align} 
DI_1&(X^N \!\rightarrow\! Y^N||Z^{N-1}) \notag \\
& \triangleq H(Y^N||Z^{N-1})-H(Y^N||X^NZ^{N-1}) \label{eq:MDI1}\\
&= \sum\limits_{n=1}^{N}D_{KL}\Big({p(y_{n}|x^{n}y^{n-1}z^{n-1})}\,||\,{p(y_{n}|y^{n-1}z^{n-1})}\Big) \label{eq:MDI2}\\
&= \sum\limits_{n=1}^{N} I(X^n;Y_n|Y^{n-1}Z^{n-1}). \label{eq:MDI3}
\end{align}
\end{defn}

\begin{proposition}[\hspace{-0.3eM}\comnew{\cite{quinn2011estimating, soltani2013inferring}}]\label{prop:impliedconn}
Assume a network as shown in Fig.~\ref{fig:broadrelay} with three nodes $\mathcal{X}$, $\mathcal{Y}$, and $\mathcal{Z}$, with corresponding random processes $X^N$, $Y^N$, and $Z^N$, respectively. Using causally conditioned Massey's directed information eliminates implied connectivity, with $DI_1(X^N \!\rightarrow\! Y^N||Z^{N-1}) = 0$ if $\mathcal{X}$~and~$\mathcal{Y}$ are not directly connected.
\end{proposition}

Both transfer entropy and Kamitake's directed information can be extended to incorporate causal conditioning from an additional random process as well. 
\comnew{
\begin{defn}[\hspace{-0.3eM} \cite{lizier2008local, lizier2010information, vakorin2009confounding}] \textit{Causally conditioned transfer entropy and sum transfer entropy are defined respectively as}
\begin{align}  
&TE(X^{n-1} \!\rightarrow Y^n||Z^{n-1}) \notag \\
&\quad\triangleq D_{KL}\Big(  {p(y_n|y^{n-1}x^{n-1}z^{n-1})} \,||\,{p(y_n|y^{n-1}z^{n-1})} \Big) \label{eq:MTE1} \\
&\quad= I(Y_n;X^{n-1}|Y^{n-1}Z^{n-1}),  \label{eq:MTE3} \\
&{TE}^\ast(X^{N-1} \!\rightarrow Y^N||Z^{N-1}) \notag \\
&\quad\triangleq  \sum\limits_{n=1}^{N} TE(X^{n-1} \!\rightarrow Y^n||Z^{n-1}). \label{eq:MTE4}
\end{align}
\end{defn}
}
\begin{defn} \label{defn:CCDI2} \textit{Causally conditioned Kamitake's directed information is defined as} 
\begin{align}\label{eq:defDI2}
DI_2(X^N\!\rightarrow\! Y^N||Z^{N-1}) \!\triangleq\! \sum\limits_{n=1}^{N} I(X_n;Y^{N}_{n+1}|X^{n-1}Y^{n}Z^{n-1}).
\end{align}
\end{defn}

Similar to Massey's directed information, causally conditioned transfer entropy and Kamitake's directed information can also be used to eliminate false inferences resulting from implied connectivity \cite{sakata2002multidimensional, vakorin2009confounding}. The proof follows along the exact same lines as Proposition\!~\ref{prop:impliedconn} and is omitted here.
\com{Also, note that despite the conditioning on a causal sequence $Z^{n-1}$, Kamitake's directed information measure is not strictly causal due to the $Y^{N}_{n+1}$ term.} 

\begin{figure}[tbp]
 \captionsetup[subfigure]{justification=centering}
 \centering
 \subfigure[]{%
    \includegraphics[scale=0.37]{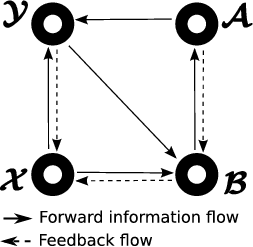}
  } 
  \\
  \subfigure[]{%
    \includegraphics[scale=0.37]{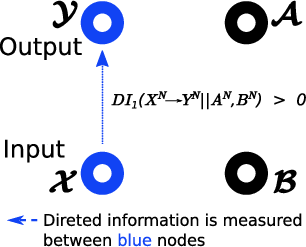}
  }
  \,\,
  \subfigure[]{%
    \includegraphics[scale=0.37]{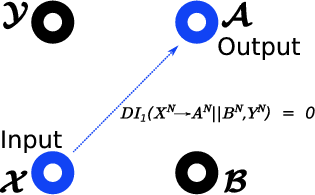}
  }
  \caption{\small{(a) An example of a multiterminal network with information flow existing between several nodes, for which we estimate the underlying causal connectivity. The solid arrows denote the direction of the forward information flow, while the dashed arrows represent feedback flow between the nodes. (b), (c) Pairwise conditional directed information calculated by choosing an input and output node, while using causal conditioning to account for the influence of all other nodes.}}  \label{fig:pairwise}
  \vspace{-4ex}
\end{figure}

Now we discuss how to apply causally conditioned directional information measures in order to estimate the functional connectivity between the ROI network. In Fig.~\ref{fig:pairwise}(a) we show an example communication network with four interacting nodes. The information flow is depicted using solid arrows and the feedback using dashed arrows, respectively. Also, some nodes do not have a direct link in between them, for example $\mathcal{X}$ to $\mathcal{A}$. In our cortical network model the nodes represent the ROIs, and the quantities $X^N$, $Y^N$, $A^N$ and $B^N$, resp., represent the random processes describing the sampled output of the EEG signals in each ROI. Our goal is then to infer the causal connectivity between the ROIs given the EEG recordings from each region.  
Towards this end we propose calculating the pairwise conditional directed information by choosing an input and output node, while using causal conditioning to account for the influence of all other nodes. For example, to estimate the connectivity from $\mathcal{X}$ to $\mathcal{Y}$ we calculate $DI_1(X^N\!\rightarrow\!Y^N||A^NB^N)$ as shown in Fig.~\ref{fig:pairwise}(b). Since there is non-zero information flow between these two nodes we expect the directed information to return a positive value. Similarly in Fig.~\ref{fig:pairwise}(c), computing $DI_1(X^N\!\rightarrow\!A^N||Y^NB^N)$ yields a zero value since there is no direct connectivity between these two nodes. Repeating this procedure pairwise for all nodes provides a functional connectivity graph representative of the directional information flow over the entire ROI network. 

\vspace{-2ex}
\subsection{Relationship between different measures}


The first relation that we are interested in is between Massey's and Kamitake's directed information measures as examined in \cite{al2008relationship, solo2008causality}. We extend this analysis to include causal conditioning on $Z^{N-1}$ and show its connection to causally conditioned MI.
\begin{proposition}\label{prop:DI2rel}
The relation between causally conditioned Massey's and causally conditioned Kamitake's directed information is given by
\begin{multline} \label{eq:DI2rel1}               
DI_2(Y^N\rightarrow X^N||Z^{N-1}) +   DI_1(X^N\rightarrow Y^N||Z^{N-1})  \\ 
			\quad	=   DI_2(X^N\rightarrow Y^N||Z^{N-1}) +   DI_1(Y^N\rightarrow X^N||Z^{N-1}).     
\end{multline}
\end{proposition}
\begin{proof} Using the definition in (\ref{eq:defDI2}) and rewriting in terms of entropy
	\begin{align} 
	&DI_2(Y^N\!\rightarrow\! X^N||Z^{N-1}) \notag \\
	&=\sum\limits_{n=1}^{N} I(Y_n;X^{N}_{n+1}|X^{n}Y^{n-1}Z^{n-1}) \\
	&= \sum\limits_{n=1}^{N}\left[ H(X^{N}_{n+1}|X^{n}Y^{n-1}Z^{n-1})-H(X^{N}_{n+1}|X^{n}Y^{n}Z^{n-1})\right] \\
	&= \sum\limits_{n=1}^{N} \left[ H(X^{N}|Y^{n-1}Z^{n-1})-H(X^{n}|Y^{n-1}Z^{n-1}) \right. \notag  \\
	& \phantom{MMM11} - \left. H(X^{N}|Y^{n}Z^{n-1})+H(X^{n}|Y^{n}Z^{n-1}) \right] \label{eq:DI2rel2} \\
	&= \sum\limits_{n=1}^{N} I(X^{N};Y_n|Y^{n-1}Z^{n-1}) - DI_1(X^N\!\rightarrow\! Y^N||Z^{N-1}),\phantom{M}
	\end{align}
	where we have used the chain rule of entropy in \eqref{eq:DI2rel2}. Rearranging yields
	\begin{align}
	DI_2&(Y^N \!\rightarrow\! X^N||Z^{N-1}) +   DI_1(X^N\!\rightarrow\! Y^N||Z^{N-1}) \notag \\ 
	&= \sum\limits_{n=1}^{N} I(X^{N};Y_n|Y^{n-1}Z^{n-1}) \label{eq:DI2rel3} \\ 
	&=\sum\limits_{n=1}^{N}\sum\limits_{i=1}^{N}I(X_i;Y_n|X^{i-1}Y^{n-1}Z^{n-1}) \label{eq:DI2rel4} \\ 
	&=\sum\limits_{n=1}^{N} I(Y^{N};X_n|X^{n-1}Z^{n-1}) \label{eq:DI2rel5}\\
	&=   DI_2(X^N\!\rightarrow\! Y^N||Z^{N-1}) +   DI_1(Y^N\!\rightarrow\! X^N||Z^{N-1}),    
	\end{align}
	where \eqref{eq:DI2rel4} follows by using the chain rule of MI, and \eqref{eq:DI2rel5} from interchanging the order of summation.
\end{proof}
 \vspace{-2ex}
\begin{defn} \label{defn:causalMI}
\textit{Causally conditioned MI is defined as the MI between $X^N$ and $Y^N$ causally conditioned on the sequence $Z^{N-1}$}
\begin{align}
I(X^N;Y^N||Z^{N-1}) &\triangleq \sum\limits_{n=1}^{N} I(X^N;Y_n|Y^{n-1}Z^{n-1}) \notag \\
&= \sum\limits_{n=1}^{N} \label{eq:causalMI} I(X_n;Y^N|X^{n-1}Z^{n-1}).  
\end{align}
\end{defn}

\comone{The following corollary then expresses causally conditioned MI in terms of directed information and follows directly from \eqref{eq:DI2rel3}-\eqref{eq:DI2rel5} in the proof of Proposition \ref{prop:DI2rel}.}
\begin{corollary} \label{prop:causalcondMI}
Causally conditioned MI is the sum of causally conditioned Massey's directed information from $X^N$ to $Y^N$, and causally conditioned Kamitake's directed information in the opposite direction:
\begin{align}
I&(X^N;Y^N||Z^{N-1}) \notag \\ &= DI_1(X^N \!\rightarrow\! Y^N||Z^{N-1}) +   DI_2(Y^N\!\rightarrow\! X^N||Z^{N-1}).
\end{align}
\end{corollary}
\comthree{There also exists a connection between between Massey's directed information and transfer entropy as shown in \cite{liu2012relationship,amblard2012relation, wibral2014directed}, which we extend and state for causal conditioning as follows.}
\begin{proposition}\label{prop:temassey}
	The relation between causally conditioned Massey's directed information and causally conditioned transfer entropy is given by
\begin{align} \label{eq:lem1} 
DI_1&(X^N \!\rightarrow\! Y^N||Z^{N-1}) \notag \\
&= \sum\limits_{n=1}^{N}\left\{ TE(X^{n-1} \!\rightarrow\! Y^n||Z^{n-1}) \right. \notag \\
& \qquad \left. + I(X_n;Y_n|X^{n-1}Y^{n-1}Z^{n-1})\right\}.
\end{align}
\end{proposition}
\comthree{We observe that the directed information is a sum over the transfer entropy and an additional term describing the conditional undirected MI between $X_n$ and $Y_n$. Massey's directed information as defined in Sec.~III-A was originally intended to measure the dependency between two $N$ length sequences. If we now relax this constraint to instead suitably measure the flow from a $N-1$ length sequence $X^{N-1}$ to a $N$ length sequence $Y^N$, while causally conditioned on $Z^{N-1}$, we then have a modified interpretation for \comnew{causally conditioned} Massey's directed information which we denote as $DI'_1$ and define as follows,
\begin{align} 
DI'_1(X^{N-1} \!\rightarrow\! Y^N||Z^{N-1}) & = {TE}^\ast(X^{N-1} \rightarrow Y^N||Z^{N-1}) \label{eq:wibral}
\end{align}
where the equality in \eqref{eq:wibral} follows directly from from \eqref{eq:MTE3}. 
Further, if we assume the sequences to be stationary and infinitely long and that the limit for $N \rightarrow \infty$ exists, then asymptotically it can be shown \cite{amblard2012relation, wibral2014directed} that the information rates for \comnew{causally conditioned Massey's modified directed information $DI'_1$ and causally conditioned transfer entropy \eqref{eq:MTE1} are in fact equal.}}

\section{Causal Bidirectional Information Flow in {EEG}}




\comone{In the following we propose an alternative bidirectional measure for estimating the causal dependency between the ROIs. This measure is motivated by the analysis of the three-terminal multiple access channel (MAC) with feedback, an important canonical building structure in networked communication.}

\vspace{-2ex}
\subsection{Causal bidirectional information (CBI)}
\comone{In order to derive the CBI measure let us first consider a preliminary result which uses causally conditioned directed information to express the information rate for such a MAC with feedback.} Fig.~\ref{fig:MAC} shows a two user MAC with feedback, with channel inputs $\mathcal{X}$ and $\mathcal{Y}$, and corresponding output $\mathcal{Z}$. The capacity region for the two user discrete memoryless MAC with feedback can be lower bounded using directed information, in a form similar to the standard cut-set bound for the MAC without feedback \cite{cover2012elements}. The information rate from $X^N$ to $Z^N$ is shown to be \cite{Kramer98}, $R_1 \leq \frac{1}{N}DI(X^N\rightarrow Z^N||Y^N)$ for all
\begin{align}\label{eq:factorization:main}
p(x_{n}&y_{n}|x^{n-1}y^{n-1}z^{n-1}) \notag \\
&= p(x_{n}|x^{n-1}z^{n-1}) \cdot p(y_{n}|y^{n-1}z^{n-1}), \quad\, n\leq N.
\end{align}
The other rate $R_2$ from $\mathcal{Y}$ to $\mathcal{Z}$ and the sum rate $R_1+R_2$ can be found in \cite{Kramer98} and are not of interest for the following discussion.

\begin{figure}[htb]
\centering
\includegraphics[scale = 0.44]{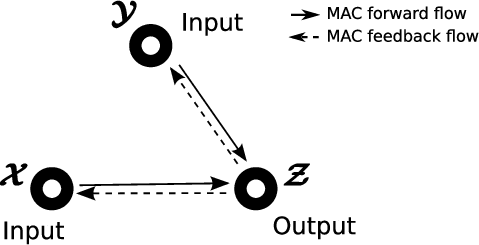}
\vspace{2ex}
\caption{A multiple access channel with feedback.}
\label{fig:MAC}
\vspace*{-1ex}
\end{figure}


\begin{figure}[tbp]
\centering \hspace{-3eM}
\subfigure[]{\includegraphics[scale = 0.45]{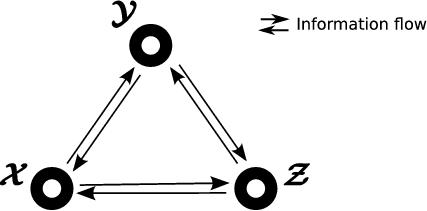}%
\label{fig_first_case}}
\\
\subfigure[]{\includegraphics[scale = 0.45]{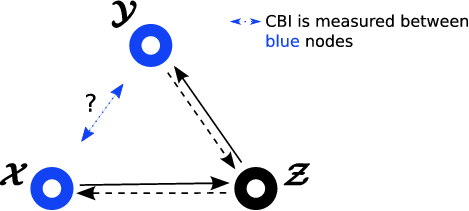}%
\label{fig_second_case}}
\caption{\small{(a) A general multi-terminal network whose connectivity we are interested in. All three communicating nodes send and receive information from each other. (b) CBI infers the connectivity between $\mathcal{X}$ and $\mathcal{Y}$ by calculating the divergence of the observed joint distribution on the network of Fig.~\!\ref{fig:CBIchannel}(a) from the one of a MAC with feedback in \eqref{eq:factorization:main}.}}
\label{fig:CBIchannel}
\vspace{-4ex}
\end{figure}

Now, consider the scenario of a general multi-terminal network as shown in Fig.~\ref{fig:CBIchannel}(a), where each node sends information and receives feedback from every other node in the network. CBI considers the three node network of Fig.~\!\ref{fig:CBIchannel}(a) and measures the information flow in between nodes  $\mathcal{X}$ and $\mathcal{Y}$ by using the MAC as a reference, as shown in Fig.~\!\ref{fig:CBIchannel}(b). In particular, we point our attention towards the relation in (\ref{eq:factorization:main}) specifying the relation between the joint distribution of the two inputs of the MAC.
The conditional independence of the inputs in the factorization of \eqref{eq:factorization:main} arises due to the causal nature of the feedback structure, where the output at the receiver $Z^{n-1}$ is available causally at both $\mathcal{X}$ and $\mathcal{Y}$ for each time $n$. For the case of a MAC with feedback there is no direct connectivity (path) between the inputs.
Any violation of \eqref{eq:factorization:main} creates dependencies between $\mathcal{X}$ and $\mathcal{Y}$, and these dependencies can be measured by the KL divergence between the joint distribution on the l.h.s. of \eqref{eq:factorization:main} and the factorization on the r.h.s. of \eqref{eq:factorization:main}. This result is summarized in the following definition.

\begin{defn} \label{defn:KL}
\comone{\textit{Consider a multi-terminal network with a source node $\mathcal{X}$, a destination node $\mathcal{Y}$, and a group of nodes $\mathcal{Z}$ interacting casually with $\mathcal{X}$ and $\mathcal{Y}$ as shown in Fig.~\!\ref{fig:CBIchannel}(a). CBI calculates the KL divergence between the induced distributions of the observed conditional distribution $p(x_ny_n|x^{n-1}y^{n-1}z^{n-1})$ and the one for an underlying MAC with feedback \eqref{eq:factorization:main}}:}
\begin{align} 
\begin{split}
I(X^N \!\leftrightharpoons\! Y^N || Z^{N-1}) = D_{KL}\Big(p(x_ny_n|x^{n-1}y^{n-1}z^{n-1})\,|| \\
\,p(x_n|x^{n-1}z^{n-1})p(y_n|y^{n-1}z^{n-1}) \Big) \end{split} \label{eq:nm1} \\
&=  \sum\limits_{n=1}^{N} \mathbb{E} \left[\log\frac{p(X_{n}Y_{n}|X^{n-1}Y^{n-1}Z^{n-1})}{p(X_{n}|X^{n-1}Z^{n-1})\, p(Y_{n}|Y^{n-1}Z^{n-1})}\right] \label{eq:nm2}  \\
&= H(X^N||Z^{N-1}) + H(Y^N||Z^{N-1}) - H(X^NY^N||Z^{N-1}).  \label{eq:nm3} 
\end{align}
\end{defn}

\comone{Therefore, CBI ascertains the \textit{direct connectivity} (as per Def.~\ref{defn:directconnect}) between two nodes in a general multi-node network and is zero if and only if the following two conditions are satisfied:}
\begin{enumerate}[label=(\roman*)]
	\item \comone{$X^n$ and $Y^n$ are independent for all $n \leq N$, i.e., there is no information flow between the two nodes.}
	
	\item \comone{There is no direct link between $\mathcal{X}$ and $\mathcal{Y}$ and all information flows only via an additional node $\mathcal{Z}$.}
\end{enumerate}

\comone{In the following proposition we show that, by definition, CBI is inherently a casual bidirectional modification of Massey's directed information.}

\begin{proposition}\label{prop:CBI}
	\comone{\textit{Causal bidirectional information (CBI) is the sum of causally conditioned Massey's directed information between $X^N$ and $Y^N$, and the sum transfer entropy in the reverse direction}:}
	
	\begin{align}
	&I(X^N \!\leftrightharpoons\! Y^N || Z^{N-1}) \notag \\*
	&\triangleq DI_1(X^N \!\rightarrow\! Y^N||Z^{N-1}) + \sum\limits_{n=1}^{N} \Big\{ TE(Y^n \!\rightarrow\! X^n||Z^{n-1}) \Big\}.
	\end{align}
	
\end{proposition} 
\begin{proof}
	We start with
	\begin{align}
	&DI_1(X^N \!\rightarrow\! Y^N||Z^{N-1}) +  {TE}^\ast(Y^{N-1} \!\rightarrow\! X^N||Z^{N-1})\notag\\
	&= \sum\limits_{n=1}^{N}\! \Big\{ I(Y_n;X^n|Y^{n-1}Z^{n-1}) + TE(Y^{n-1} \!\rightarrow\! X^n||Z^{n-1}) \Big\}, \label{eq:nm_phi}
	\end{align}
	where the equality in \eqref{eq:nm_phi} follows from (\ref{eq:MDI3}) in Definition~\ref{defn:Masseycausal}. Denoting the term inside the summation on the r.h.s. of (\ref{eq:nm_phi}) by $\Phi(\cdot)$ and rewriting yields 
	\begin{align} 
	&\Phi(x^n,y^n,z^{n-1}) \notag \\
	&= I(Y_n;X^n|Y^{n-1}Z^{n-1}) + TE(Y^{n-1} \!\rightarrow\! X^n||Z^{n-1}) \label{eq:phi1} \\
	&= \mathbb{E}\left[ \log \frac{p(Y_n|X^nY^{n-1}Z^{n-1})}{p(Y_n|Y^{n-1}Z^{n-1})}\right] \notag \\
	& \qquad  + \mathbb{E}\left[ \log \frac{p(X_n|Y^{n-1}X^{n-1}Z^{n-1})}{p(X_n|X^{n-1}Z^{n-1})} \right]\label{eq:phi2} \\
	&= \mathbb{E}\left[ \log \frac{p(Y_nX_n|Y^{n-1}X^{n-1}Z^{n-1})}{p(Y_n|Y^{n-1}Z^{n-1})p(X_n|X^{n-1}Z^{n-1})} \right], \label{eq:phi4} 
	\end{align}
	where in \eqref{eq:phi2} we have made use of \eqref{eq:MDI2} and \eqref{eq:MTE1}, respectively, and \eqref{eq:phi4} follows from the chain rule of joint probability
	\begin{align}
	p&(y_nx_n|x^{n-1}y^{n-1}z^{n-1}) \notag \\
	&= p(y_n|x_n x^{n-1}y^{n-1}z^{n-1})\cdot p(x_n|x^{n-1}y^{n-1}z^{n-1}). \label{eq:jointprob1}
	\end{align}
	Taking the summation on \eqref{eq:phi4} and comparing with \eqref{eq:nm2} proves the claim.
\end{proof} 

\vspace{-2ex}
\begin{corollary}
	CBI is a symmetric measure \comnew{
	\begin{align}
	&I(X^N \!\leftrightharpoons\! Y^N || Z^{N-1}) \notag \\
	&= DI_1(X^N \!\rightarrow\! Y^N||Z^{N-1}) +  {TE}^\ast(Y^{N-1} \!\rightarrow\! X^N||Z^{N-1})\\
	&=  DI_1(Y^N \!\rightarrow\! X^N||Z^{N-1}) + {TE}^\ast(X^{N-1} \!\rightarrow\! Y^N||Z^{N-1}).
	\end{align}}
\end{corollary}
\begin{proof} Without any loss of generality, the joint probability distribution in \eqref{eq:jointprob1} can alternatively be expanded as
	\begin{align} \label{eq:jointprob2}
	p&(y_nx_n|y^{n-1}x^{n-1}z^{n-1}) \notag \\
	&= p(x_n|y_n y^{n-1}x^{n-1}z^{n-1})\cdot p(y_n|y^{n-1}x^{n-1}z^{n-1}).
	\end{align} 
	Using the above equation and rewriting \eqref{eq:phi4} then gives us
	\begin{align}
	& \Phi(x^n,y^n,z^{n-1}) \notag \\
	&= \mathbb{E}\left[ \log \frac{p(Y_nX_n|Y^{n-1}X^{n-1}Z^{n-1})}{p(Y_n|Y^{n-1}Z^{n-1})p(X_n|X^{n-1}Z^{n-1})} \right] \\
	&= \mathbb{E}\left[ \log \frac{p(X_n|Y^nX^{n-1}Z^{n-1})}{p(X_n|X^{n-1}Z^{n-1})}\right] \notag \\
	& \qquad + \mathbb{E}\left[ \log \frac{p(Y_n|Y^{n-1}X^{n-1}Z^{n-1})}{p(Y_n|Y^{n-1}Z^{n-1})} \right] \\
	& = I(X_n;Y^n|X^{n-1}Z^{n-1}) + TE(X^{n-1} \!\rightarrow\! Y^n||Z^{n-1}). \label{eq:corr_sym}
	\end{align}
	Taking the sum on both sides of \eqref{eq:corr_sym} yields the required result.
\end{proof}

We also evaluate the expression for CBI by comparing it to conditional MI. Conditional MI measures the divergence between the actual observations and those which would be observed under the Markovian assumption ${X^N \!\leftrightarrow\! Z^{N-1} \!\leftrightarrow\! Y^N}$,
\begin{align}
& I(X^N ;Y^N|Z^{N-1}) \notag \\
&=  \mathbb{E}\left[\log\frac{p(X^{N}Y^{N}|Z^{N-1})}{p(X^{N}|Z^{N-1})\, p(Y^{N}|Z^{N-1})}\right] \label{eq:lemma1}\\
&=  \sum\limits_{n=1}^{N} \mathbb{E}\left[\log\frac{p(X_{n}Y_{n}|X^{n-1}Y^{n-1}Z^{N-1})}{p(X_{n}|X^{n-1}Z^{N-1})\, p(Y_{n}|Y^{n-1}Z^{N-1})}\right] \label{eq:lemma2} \\
 &= H(X^N|Z^{N-1}) + H(Y^N|Z^{N-1}) - H(X^NY^N|Z^{N-1}), \label{eq:lemma3} 
\end{align}
where \eqref{eq:lemma2} follows from the chain rule of probability.
Conditional MI is zero if and only if $X^N$ and $Y^N$ are conditionally independent given $Z^{N-1}$. \comone{By comparing conditional mutual information \eqref{eq:lemma2} with the expression for CBI in \eqref{eq:nm2}, we notice that CBI uses causal conditioning as $Z^{N-1}$ is replaced with $Z^{n-1}$.}

\vspace{-1ex}
\section{\comnew{Inferring Change In Functional Connectivity}}
\subsection{Preliminaries}

In our analysis, we separately extract the EEG response sections for each of the two audio quality levels and calculate the information measures individually for each of them. This allows us to compare how the different probability distributions used in each of the presented information measures effect the ability to detect a change in information flow among the ROIs in Fig.~\ref{fig:ROI}, between the cases when the subjects listen to high quality audio as opposed to degraded quality audio.

We begin by selecting a source and a destination ROI, $\mathcal{X}$ and $\mathcal{Y}$, respectively. The other six remaining ROIs are considered to represent the side information $\mathcal{Z}$. Since all electrodes in a ROI are located within close proximity of one another and capture data over the same cortical region, we consider every electrode in an ROI as an independent realization of the same random process. For example, the sampled EEG data recorded on every electrode within region $\mathcal{X}$, in a given time interval $N$, is considered a realization of the same random process $X^N$. This increases the sample size for the process $X^N$, reducing the expected deviation between the obtained empirical distribution of $X^N$ and the true one. \comnew{The discussed information measures are therefore calculated between the ROI pairs (and not between individual electrodes) for a total of 2-permutations of 8 ROIs, i.e., 56 combinations of source-destination pairs.}

In our earlier work \cite{mehtaMBMC} we demonstrated that the output EEG response to the audio quality, over an ROI, converges to a Gaussian distribution with zero mean. The intuition here is that the potential recorded at the EEG electrode at any given time-instant can be considered as the superposition of responses of a large number of neurons. Thus, the distribution of a sufficiently high number of these trials taken at different time instances converges to a Gaussian distribution as a result of the Central Limit Theorem. Fig.~\ref{fig:hist} shows the histogram of the sampled EEG data, formed by concatenating the output over all sensors in one ROI for a single subject. The sample skewness and kurtosis of the EEG output distribution are shown in Table~\ref{tab:skew_kurt}. For a Gaussian distribution the skewness equals 0 and the kurtosis equals 3 \cite{corder2014nonparametric, kim2013statistical}. To test Gaussanity of a large sample set, the sample skewness and kurtosis should approach these values, while an absolute skew value larger than 2 \cite {kim2013statistical} or kurtosis larger than 7 \cite{west1995structural} may be used as reference values for determining substantial non-normality. By inspecting the sample estimates in Table~\ref{tab:skew_kurt} and comparing the histogram to a Gaussian distribution in Fig.~\ref{fig:hist}, we observe that the EEG output distribution is indeed strongly Gaussian.

Knowing that the interacting random processes from the ROIs converge to a Gaussian distribution \cite{mehtaMBMC} allows us to formulate analytical closed-form expressions for calculating the information measures. The joint entropy of a $n$-dimensional multivariate Gaussian distribution with probability density $p(z_1 \ldots z_n)$ is known to be given by \cite{cover2012elements}
\begin{align} \label{eq:GaussEntropy}
H(Z_1 \ldots Z_n) = \frac{1}{2}\log{ (2\pi e)^n |\mathsf{C}(Z_1 \ldots Z_n)|  },
\end{align}
where $\mathsf{C(\cdot)}$ is the covariance matrix and $|\cdot |$ is the
determinant of the matrix. 

If $X^N$, $Y^N$, and $Z^N$ are jointly Gaussian distributed, then using \eqref{eq:GaussEntropy} in conjunction with \eqref{eq:nm3} reduces CBI to a function of their joint covariance matrices, 
\begin{align} \label{eq:Icov}
I(X^N  \!\leftrightharpoons\! Y^N ||& Z^{N-1}) = \frac{1}{2}\sum\limits_{n=1}^{N}\log \Bigg\{ \frac{|\mathsf{C}(X^{n-1}Y^{n-1}Z^{n-1})|}{|\mathsf{C}(X^nY^nZ^{n-1})|} \notag \\
&\cdot\frac{|\mathsf{C}(X^{n}Z^{n-1})|\cdot|\mathsf{C}(Y^{n}Z^{n-1})|}{|\mathsf{C}(X^{n-1}Z^{n-1})|\cdot|\mathsf{C}(Y^{n-1}Z^{n-1})|} \Bigg\}	. 
\end{align} 
In a similar manner, causally conditioned Massey's directed information, Kamitake's directed information, and sum transfer entropy, resp., can be reduced to obtain the following expressions, 
\begin{align} \label{eq:DI1cov}
&DI_1(X^N \!\rightarrow\! Y^N ||Z^{N-1}) \notag \\ 
&= \frac{1}{2}\sum\limits_{n=1}^{N}\log{ \frac{|\mathsf{C}(Y^nZ^{n-1})|\cdot|\mathsf{C}(X^{n}Y^{n-1}Z^{n-1})|} {|\mathsf{C}(Y^{n-1}Z^{n-1})|\cdot|\mathsf{C}(X^n Y^nZ^{n-1})|}	}, \quad\,\,
\end{align} 
\begin{align} \label{eq:DI2cov}
&DI_2(X^N \!\rightarrow\! Y^N ||Z^{N-1}) \notag \\ 
&= \frac{1}{2}\sum\limits_{n=1}^{N} \!\log{ \frac{|\mathsf{C}(X^{n-1}Y^{N}Z^{n-1})|\cdot|\mathsf{C}(X^{n}Y^{n}Z^{n-1})|} {|\mathsf{C}(X^{n-1}Y^{n}Z^{n-1})|\cdot|\mathsf{C}(X^n Y^N Z^{n-1})|}	} ,
\end{align}
\comnew{
\begin{align} \label{eq:TEcov}
&{TE}^\ast(X^{N-1} \!\rightarrow\! Y^N ||Z^{N-1}) \notag \\
&= \frac{1}{2}\sum\limits_{n=1}^{N} \!\log{ \frac{|\mathsf{C}(Y^nZ^{n-1})|\cdot|\mathsf{C}(Y^{n-1}X^{n-1}Z^{n-1})|} {|\mathsf{C}(Y^{n-1}Z^{n-1})|\cdot|\mathsf{C}(Y^n X^{n-1}Z^{n-1})|}.	}  
\end{align}} 

\begin{figure}[tbp]
	\centering
	\vspace{-2ex}
	\includegraphics[width = \columnwidth]{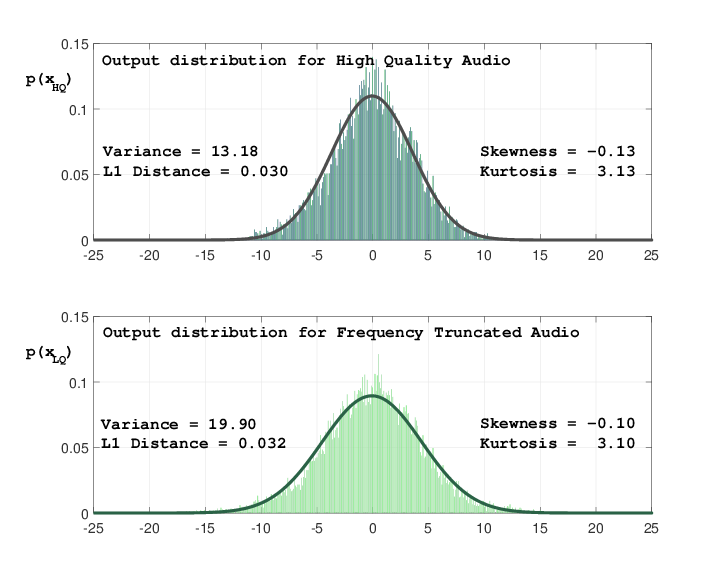}
	\caption{\small{\comone{Sampled EEG output distributions over a single ROI for high quality and frequency distorted audio input-stimulus, respectively. The data used to construct the histogram is the combined output over all sensors in ROI~5 during the period of a single frequency-truncated trial, for Subject S1. The Gaussian fit is obtained by using an estimator that minimizes the $L_1$ distance between the fitted Gaussian distribution and the histogram data.}}
	} \label{fig:hist}
	\vspace{-1ex}
\end{figure}

\newcolumntype{L}[1]{>{\raggedright\let\newline\\\arraybackslash\hspace{0pt}}m{#1}}
\newcolumntype{C}[1]{>{\centering\let\newline\\\arraybackslash\hspace{0pt}}m{#1}}
\newcolumntype{R}[1]{>{\raggedleft\let\newline\\\arraybackslash\hspace{0pt}}m{#1}}

\renewcommand{\arraystretch}{1.2}
\begin{table}[htbp]\centering
	\small
	\caption{ \small{\comone{Sample skewness and kurtosis estimated for a single subject S1 using data from a single frequency-truncated trial. The sampled EEG data for a given ROI was obtained by concatenating the output of all electrodes in that ROI.}}}
	\begin{tabular}{@{}c|C{9mm} C{9mm} C{9mm} |C{9mm} C{9mm} C{9mm} @{}}
		\toprule
		&  & HQ &  &  & LQ & \\
		ROI & Mean & Skewness & Kurtosis & Mean & Skewness & Kurtosis\\
		\midrule
		1 	& -0.003 	&	0.180	&	2.820 		& -0.014 	&	0.187	&	2.816\\
		2 	& 	0.148 	&	0.057	&	2.943		& 0.159 	&	0.214	&	2.786\\
		3 	& 	0.027 	&	0.069	&	2.931		& 0.014 	&	0.054	&	2.945\\
		4 	& 	0.216 	&	-0.154	&	3.155		& 0.2975	&	0.131	&	2.869\\
		5 	& -0.058 	&	-0.132	&	3.133		& -0.073 	&	-0.103 &	3.103\\
		6 	& -0.305 	&	-0.145	&	3.146		& -0.283 	&	-0.259	&	3.260\\
		7 	& -0.111 	&	 0.016	&	2.984		& -0.101 	&	 0.050	&	2.950\\
		8 	& 	0.465	&	-0.070	&	3.071		& 0.445	&	-0.342	&	3.342\\
		\bottomrule
	\end{tabular}
	\label{tab:skew_kurt}
	\vspace{-1ex}
\end{table}
\vspace*{-2ex}
\subsection{\com{Receiver operating characteristic (ROC) curves}}
In order to evaluate the accuracy with which an information measure and in particular CBI can distinguish between the perceived audio quality we conduct a receiver operating characteristic (ROC) curve analysis \cite{fawcett2006introduction, krzanowski2009roc} on the generated vectors of measurements for the high and degraded quality audio, respectively. The ROC curve serves as a non-parametric statistical test to compare different information rates \cite{bossomaier2016transfer, ito2011extending, garofalo2009evaluation, lizier2013inferring} and has the advantage that a test statistics can be generated from the observed measurements.

Consider a general binary classification scheme between two classes
$\mathcal{P}$ and $\mathcal{N}$, labeled as \textit{positive} and
\textit{negative}, respectively. These classes contain samples from the
respective class distributions with $|\mathcal{P}|=N_p$ and
$|\mathcal{N}|=N_n$. Also, every individual instance from the positive class is associated with a
known measurement score $S^{(p)}_i$ which is a random variable with an
unknown distribution and corresponding values
$s^{(p)}_i,\, i=1,\ldots ,N_p$. Similarly, for the negative class we have
the random variable $S^{(n)}_i$ with values $s^{(n)}_i,\, i=1,\ldots ,N_n$.

For a given discrimination threshold value $T$, the classification rule is such that the individual is allocated to the positive class if its score exceeds the threshold, i.e., $s^{(p)}_i > T$. Else, $s^{(p)}_i \leq T$, then the instance is allocated to the negative class. The true positive probability \cite{krzanowski2009roc} is the probability that an individual from the positive class is correctly classified as belonging to the positive class:
\begin{align}
P_{tp}(T) = \frac{\sum_{i=1}^{N_p} \mathbbm{1}( s^{(p)}_i > T)}{N_p}, \label{eq:ptp}
\end{align}
where $\mathbbm{1}(\cdot)$ denotes the indicator function. Likewise, the false positive probability is the probability that an individual from the negative population is misclassified (incorrectly allocated) as belonging to the positive population:
\begin{align}
P_{fp}(T) = \frac{\sum_{i=1}^{N_n} \mathbbm{1}( s^{(n)}_i > T)}{N_n}. \label{eq:pfp}
\end{align}

Since in general the best choice for $T$ is not known, the measurement scores themselves often serve as optimal choices for the discrimination threshold values \cite{fawcett2006introduction}.
The ROC curve is then a graphical representation obtained by plotting the true positive probability $P_{tp}(T)$ versus the false positive probability $P_{fp}(T)$ as a function of the discrimination threshold value $T$.

The area under the ROC curve, denoted using the variable $\theta$, summarizes the results over all possible values of the threshold $T$ \cite{fawcett2006introduction, bradley1997use}, with possible values ranging from $\theta=0.5$ (no discriminative ability) to $\theta=1.0$ (perfect discriminative ability). \comnew{In fact, the area under the curve is shown to be equivalent to the Wilcoxon-Mann-Whitney non-parametric test-statistic \cite{corder2014nonparametric} and can be used to determine whether samples selected from the two class populations have the same distribution \cite{hanley1982meaning, delong1988comparing}.} \comnew{More specifically, the area under the ROC curve is equivalent to the probability that a randomly selected instance from the positive class will have a measurement score that is greater than a randomly selected instance from the negative class.}


\vspace{-2ex}
\subsection{Performance of information measures on simulated EEG data}
To obtain an initial assessment of the information measures introduced in Sec.~IV and Sec.~V, we apply these measures to synthetic EEG data. The advantage of this approach is that the true causal connectivity graph of the simulated network is known and that, in contrast to real EEG data obtained from human subjects, in principle there is no limitation to the number of trials. The EEG data is created using a wavelet transform following the approach described in \cite{bridwell2016spatiospectral}, wherein the EEG is decomposed as a convolution of a series of basis functions (i.e., wavelets) within selected frequency bands listed in Table~\ref{tab:EEGsimulated}.
In \cite{bridwell2016spatiospectral} the probability distribution of the wavelet coefficients for each frequency band was estimated using real human EEG data sampled at 250Hz, which was shown to be a logistic distribution with heavier tails than a Gaussian distribution. We now generate simulated wavelet coefficients for a particular frequency band by multiplying the associated logistic distribution with a constant scaling factor and then randomly drawing samples from this scaled distribution. The simulated EEG signal is then obtained by an inverse wavelet transform of these randomly drawn samples for all frequency bands in Table~\ref{tab:EEGsimulated}. In this work, two different sets of scaling factors are used to simulate change in network connectivity. Since the magnitude of the scaling coefficient controls the spectral energy (activity) in that band, we here use the activity labels ``high" and ``low", respectively, to differentiate between these two sets of connectivities as shown in Table~\ref{tab:EEGsimulated}.
For simulation purposes, the wavelet-coefficients for high activity in the delta band (0-3.91 Hz) and theta band (3.91-7.81 Hz) are scaled by a factor chosen from a distribution $\mathcal{U}(0.9, 1)$, where $\mathcal{U}(a,b)$ denotes the uniform distribution within the interval $(a,b)$. Low activity is simulated by choosing the scaling factors for the delta and theta bands from $\mathcal{U}(0.55, 0.85)$. Note that the uniform distribution and the associated intervals are only chosen to demonstrate a change in delta and theta activity levels, and in theory any suitable set of high-valued and low-valued scaling factors can potentially be used to this effect.  

A source and sink electrode location was assigned to each frequency band \cite{bridwell2016spatiospectral} as shown in Table~\ref{tab:EEGsimulated}. The locations of these electrodes are marked in red in Fig.~\ref{fig:ROI}. Simulated EEG signals are generated across the scalp by spherical spline interpolation across neighboring electrodes by using the \textit{eeg\_interp} function in EEGLAB \cite{delorme2004eeglab}. In essence this function simulates EEG scalp topography by smearing electric potentials through volume conduction \cite{nunez2006electric}. This results in a synthetic EEG waveform at each electrode with energy within each of the four characteristic frequency bands.

\newcolumntype{L}[1]{>{\raggedright\let\newline\\\arraybackslash\hspace{0pt}}m{#1}}
\newcolumntype{C}[1]{>{\centering\let\newline\\\arraybackslash\hspace{0pt}}m{#1}}
\newcolumntype{R}[1]{>{\raggedleft\let\newline\\\arraybackslash\hspace{0pt}}m{#1}}

\renewcommand{\arraystretch}{1.2}
\begin{table}[btp]\centering
	\small
	\caption{ \small{Wavelet coefficient scale factor for high and low activity for each source-sink pair.}}
	\vspace{1ex}
	\begin{tabular}{@{}c c | C{18mm} |C{13mm} C{20mm} @{}}
		\toprule
		Source & Sink & Freq. (Hz) & High & Low \\
		\midrule
		Fz 	& Cz 	&	0 - 3.91 (delta)	&	$\mathcal{U}$(0.9, 1) 		& $\mathcal{U}$(0.55, 0.85) 	\\
		\midrule
		F5, F6 	& 	P5, P6 	&	3.91 - 7.81 (theta)	&$\mathcal{U}$(0.9, 1)		& $\mathcal{U}$(0.55, 0.85) \\
		\midrule
		P6, F5 	& 	Cz 	&	7.81 - 15.62 (alpha) &	1		& 1	\\
		\midrule
		T7 	& 	T8 	&	15.62 - 31.25	(beta)&	1		& 1 \\
		\bottomrule
	\end{tabular}
	\label{tab:EEGsimulated}
	\vspace{-4ex}
\end{table}

We generate 1000 trials of simulated EEG data for each of the two cases of high and low activity. For each trial, a total of 430500 samples (at 250 Hz) are generated for each of the electrodes listed in Table~\ref{tab:EEGsimulated}. The EEG output generated at each of these electrodes is tested for Gaussianity similar to the process described in Sec.VI-A, including using an estimator that minimizes the $L_1$ distance between the fitted Gaussian distribution and the generated data, and indeed verified to be near Gaussian. For each trial of high and low activity the source electrode (of a given frequency band for which the activity changed) is designated as $\mathcal{X}$, the corresponding destination electrode as $\mathcal{Y}$, and all other remaining electrodes in Table~\ref{tab:EEGsimulated} are grouped together as $\mathcal{Z}$. The sample space for each trial is obtained by dividing the 430500 samples into 14350 sections each of block length $N=30$ (120 ms). \comnew{We then compute the different information measures \eqref{eq:Icov}-\eqref{eq:TEcov} over the block length $N$ and compare how the different probability distributions used in each measure effect the ability to detect change in connectivity.} The corresponding covariance matrices, $\forall \,\,n = 1,\ldots,N\,$, are estimated from the 14350 sections, assuming stationarity of the EEG signals within each section. \comnew{While here we always sum up to a depth of $N=30$, we note that the estimate accuracy of the information measures, the computational time and complexity, and the probability of a false negative can all be further improved by rigorously selecting optimal
history lengths\cite{vicente2011transfer, wibral2014directed}.}

The ROC curve analysis is used to test the accuracy with which these information measures can distinguish between the cases of high and low activity. The positive class $\mathcal{P}$ is constructed by concatenating the information rates for the high activity, and likewise the negative class $\mathcal{N}$ for low activity. The ROC curve for the information measure is then constructed using the empirical probabilities $P_{tp}(T)$ and $P_{fp}(T)$ calculated according to \eqref{eq:ptp} and \eqref{eq:pfp} by varying the threshold $T$ over all possible values of observed information rates. The resulting ROC curves are shown in Fig.~\ref{fig:FzCz}. \comnew{We observe that CBI consistently shows a significant variation in connectivity performing either better ($F5\!\leftrightarrow\!P6$) or at least identical ($F_z\!\leftrightarrow\!C_z$) to the other casually conditioned measures, with the larger area under the ROC curve for CBI indicating superior classification accuracy.
Also, since there is no direct link in the direction from $Cz\!\rightarrow\!Fz$ (see Table~\ref{tab:EEGsimulated}), the causally conditioned directional information measures do not show significant discriminability in connectivity. 
However, an important observation here is that for the electrode pair F5 and P6, the directional measures identify an almost symmetric change in connectivity in both directions. This is because for complex interconnected systems (as modeled by the EEG electrodes in Table~\ref{tab:EEGsimulated}) there is not necessarily a greater directional information transfer from source to destination than from destination to source \cite{schreiber2000measuring, pulkitgrover2015}. In \cite{schreiber2000measuring, wibral2014directed, pulkitgrover2015}, the authors specifically highlight the inherent limitations of applying directed information to infer causal connectivity in an unknown network without \textit{a priori} knowledge of which node is the source and which node is the destination.}


\begin{figure}[htb]
	\centering
	\includegraphics[width=\columnwidth]{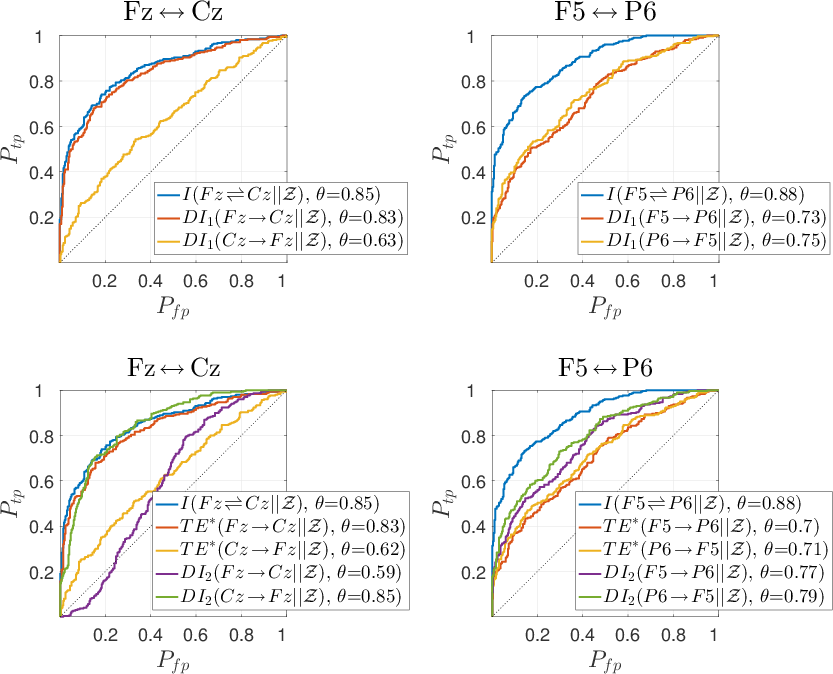}
	\vspace{0.5ex}
	\caption{\small{ROC curves classifying the change in connectivity between high and low activity for the synthetic EEG model.}}
	\label{fig:FzCz}
	\vspace*{-4ex}
\end{figure}


\vspace{-2ex}
\subsection{\com{Performance of information measures on real EEG brain data}}

\comtwo{In the following, we show that CBI offers a significant improvement in performance over the other information measures when applied to infer connectivity for real EEG data from human subjects. \comnew{Towards this end we select a source ROI $\mathcal{X}$, a destination ROI $\mathcal{Y}$, and group all the other six remaining ROIs as $\mathcal{Z}$. We then use the closed form expressions \eqref{eq:Icov}-\eqref{eq:TEcov} to calculate the instantaneous information rates  for each of the 56 possible combinations of source-destination ROI pairs. The corresponding covariance matrices $\mathsf{C(\cdot)}$ of the joint Gaussian distributions for each source-destination combination are estimated} using 125 millisecond long overlapping sliding windows (having a total length of $N=32$ points, with a 8 point overlap on each side) of the trial data. We assume stationarity of the EEG signals within these windows, implying that the functional connectivity does not vary in this segment and also make an implicit Markovian assumption that the current time window captures all past activity between the electrodes.
For each sliding window position the covariance matrix is computed as an average over the whole sample space for that specific subject, $\forall \,\,n = 1,\ldots,N$.
The sample space for calculating each covariance matrix is created by pooling across segments with the same audio quality from a total of 56 trials across multiple presentations of different distortion and music types.
Further, since we consider each electrode in an ROI to be an independent realization of the same random process, we also pool across all electrodes in the ROI. \comnew{Thus, for 9 electrodes in each ROI and 56 trials with the music quality repeating twice in each trial, we obtain a total of $56 \times 2 \times 9 = 1008$ sections (realizations) for each window position per ROI. In order to try and ensure that the number of electrodes in an ROI does not change bias w.r.t. to the increasing sample size, we always use the same number of pooled sections (realizations) to create the effective sample space.} The covariance is then computed using sample realizations at the corresponding time across all 1008 sections, assuming stationarity at that time point across sections. 
Thus, over all sliding window positions we obtain a vector 
\begin{align}
\mathbf{I}&(X^N \!\leftrightharpoons\! Y^N || Z^{N-1}) \notag \\ 
&= [I_1(X^N \!\leftrightharpoons\! Y^N || Z^{N-1}), \dots , I_K(X^N \!\leftrightharpoons\! Y^N || Z^{N-1})]
\end{align}
where $K$ is the total number of window positions and $I_k(X^N \!\leftrightharpoons\! Y^N || Z^{N-1})$ is the instantaneous information rate for the $k$-th window computed using \eqref{eq:Icov}. The vectors ${\mathbf{DI}_1(X^N \!\rightarrow\! Y^N ||Z^{N-1})}$, ${\mathbf{DI}_2(X^N \!\rightarrow\! Y^N ||Z^{N-1})}$ and ${\mathbf{TE}^\ast(X^{N-1} \!\rightarrow\! Y^N ||Z^{N-1})}$ are constructed similarly. These vectors will be considered further and are hereby denoted as $\mathbf{I}$, $\mathbf{DI}_1$, $\mathbf{DI}_2$, and $\mathbf{TE}^\ast$ for brevity.
}

\com{Also, we only select a subset of 8 subjects who provide the largest median mutual information values on the event related potential (ERP) channel connecting the audio stimulus and the quantized EEG sensor outputs (see \cite{mehtaMBMC} for details). It follows from the discussion in \cite{mehtaMBMC} that these were the subjects whose EEG recordings showed the maximum response to the changing audio quality levels. We do not combine data (sample space) between subjects, and the information rates are calculated separately for each of these 8 subjects for all possible combinations of ROI pairs.}


\subsubsection{Instantaneous information results}



Fig.~\ref{fig:Massey} illustrates the vectors for the causal measures $\mathbf{I}$, $\mathbf{DI}_1$, and $\mathbf{TE}^\ast$ plotted over time for six different combinations of ROI pairs for a second of EEG trial data, averaged over all eight subjects. The 0\,ms marker on the horizontal axis is the stimulus onset time, i.e., the instant when the audio quality changed.
The results indicate a notable difference between the amount of causal information flow for high and degraded audio, in particular for CBI. More specifically, there appears to be a higher amount of information flow between the ROIs when the subject listens to degraded quality audio.

\begin{figure*}
	\centering
	\includegraphics[width=\textwidth]{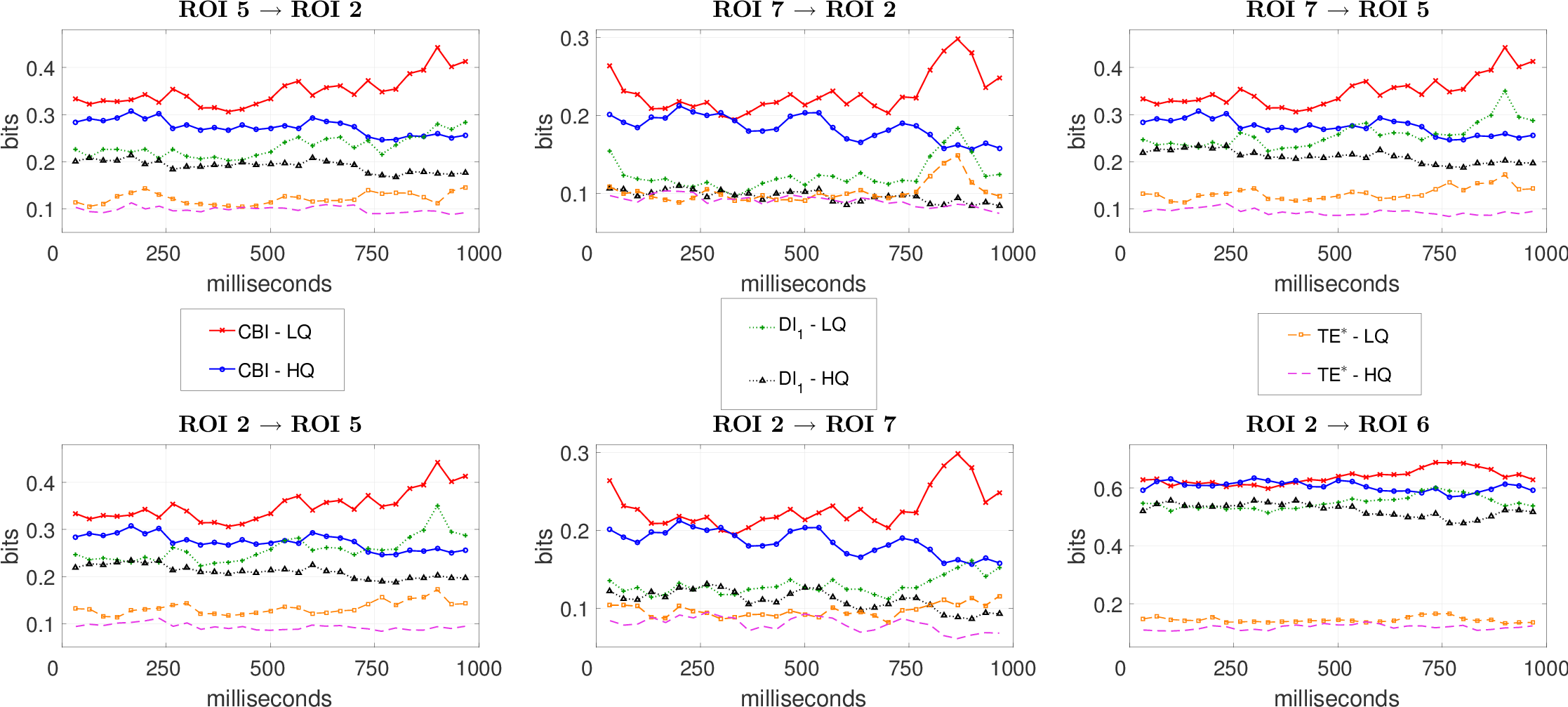}
	\vspace{1ex}
	\caption{\small{Instantaneous information transfer rate vectors $\mathbf{I}$, $\mathbf{DI}_1$ and $\mathbf{TE}^\ast$ between six different ROI pairs, averaged over eight human subjects for a second of the trial data after stimulus onset. We observe a significant difference between the transfer rates of the two audio qualities, with the transfer rates being notably higher across all ROIs when the subjects listen to degraded quality audio.}
	} \label{fig:Massey}
	\vspace{-4ex}
\end{figure*}

\comtwo{The ROIs we select for plotting in Fig.~\ref{fig:Massey} are based on the direction and order of the auditory sensing pathway in the brain. The primary auditory cortex located in the left and right temporal lobes is the first region of the cerebral cortex to receive auditory input. The higher executive functions and subjective responses are a result of the information exchange between the primary auditory cortex and the other cortical regions, predominantly including the prefrontal cortex \cite{romanski1999dual, rauschecker2009maps}. We therefore plot the information rates between the temporal lobes (ROI~5, ROI~7), and the temporal lobes and the prefrontal cortex (ROI~2) in Fig.~\ref{fig:Massey}. \comnew{We observe here that $\mathbf{DI}_1$ and $\mathbf{TE}^\ast$ detect an almost identical (symmetric) information flow in both directions between ROI pairs $2, 5$ and $2, 7$, respectively.}
Other important auditory pathways are the dual dorsal and ventral streams which carry information medially between the prefrontal cortex, the temporal lobes and the parietal lobe \cite{ahveninen2006task, wang2008dual}. We therefore also include the rates from the prefrontal cortex to the parietal lobe (ROI~2$\rightarrow$ROI~6) in Fig.~\ref{fig:Massey}.}

Note that the first second after stimulus onset can be considered as the \textit{transient response} of the brain's initial perception towards stimulus. Since we are interested in detecting change in connectivity in response to the change in audio quality, for the rest of our analysis here we focus on the information transfer over this initial one second transient period.

\subsubsection{\com{ROC results}}
We construct separate classes for each of the different information measures for every ROI pair. Therefore for a given ROI pair, $\mathcal{P}$ corresponds to one of the $\mathbf{I}$, $\mathbf{DI}_1$, $\mathbf{DI}_2$, or $\mathbf{TE}^\ast$ rate vectors for degraded quality concatenated for all eight subjects over the first second after stimulus onset, and likewise for $\mathcal{N}$ with respect to high quality. Then, the observed information rates in each of the vectors are equivalent to the measurement scores.

In our setup, given a discrimination threshold value $T$ for the information rate, $P_{tp}(T)$ is calculated according to \eqref{eq:ptp} and corresponds to the empirical probability that the observed information rate for degraded quality audio is correctly classified as degraded quality, while $P_{fp}(T)$ is calculated according to \eqref{eq:pfp} and is the empirical probability that an observed information rate for high quality audio is misclassified as degraded quality. 

\begin{figure*}
	\centering
	\includegraphics[scale = 0.4]{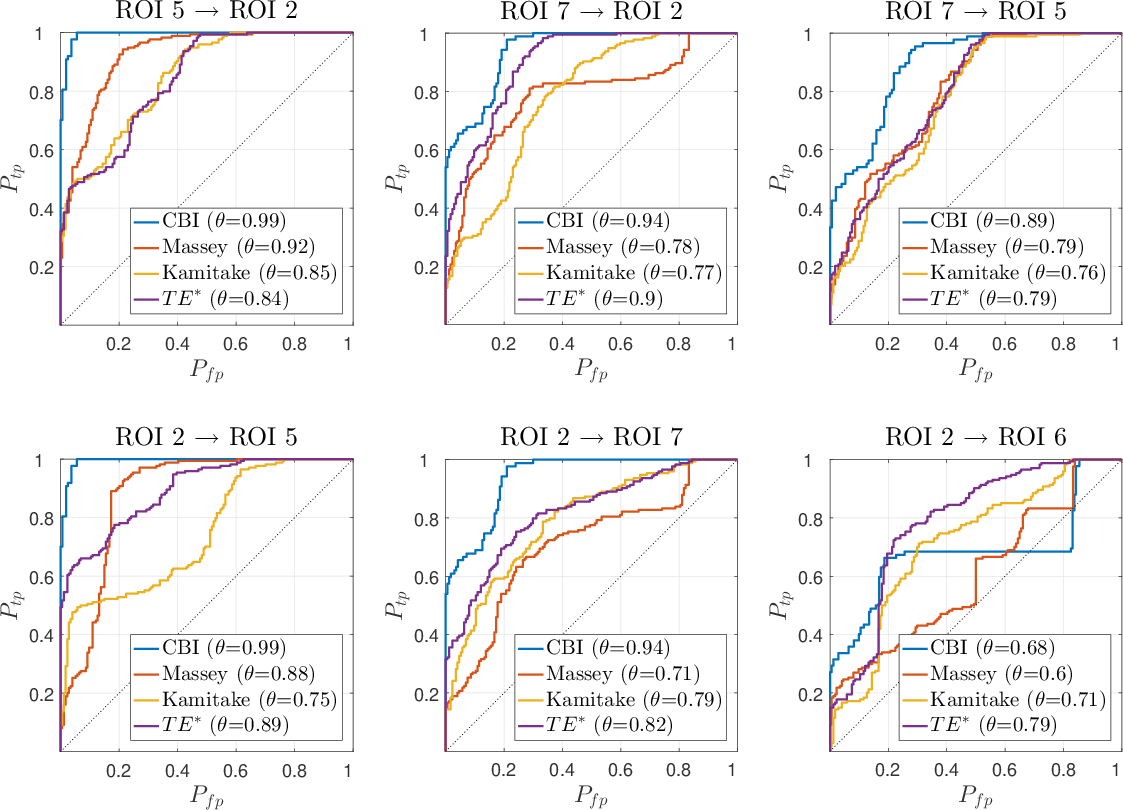}
	\vspace{1ex}
	\caption{\small{\com{ROC curves for different information measure and the same six ROI pairs shown in Fig.~\ref{fig:Massey}, along with the corresponding area $\theta$ under the ROC curve. The ROC curve for each information measure is constructed separately by using the corresponding information rates over eight human subjects.}	}} \label{fig:ROC}
	\vspace{-4ex}
\end{figure*}

Fig.~\ref{fig:ROC} shows the ROC curves for each of the information measures, for the same six ROI pairs shown in Fig.~\ref{fig:Massey}. For each of the constructed ROC curves in Fig.~\ref{fig:ROC} we also provide the total area under that curve $\theta$ \cite{delong1988comparing, bradley1997use}. We observe from the results in Fig.~\ref{fig:ROC} that the classifier performance of the information measures varies considerably depending on the ROI pair chosen. This is due to the fact that substantial changes are likely to occur between the observed information rates for high quality and degraded quality only for those ROI connections which are actively involved in detecting and processing an auditory stimulus response.
For five out of the six ROI pairs depicted in Fig.~\ref{fig:ROC}, we notice that CBI performs exceedingly well and has the best discriminability among all given information measures. 

\subsubsection{\com{Statistical significance testing of area under the ROC curve}}
The precision of an estimate of the area under a ROC curve is validated by conducting a test for statistical significance \cite{cortes2005confidence, liu2005testing}. The statistical significance test is used to determine whether or not the area under the ROC curve is lesser than a specific value of interest $c$. The null hypothesis that we are interested in testing can therefore be defined as 
\begin{align}
H_0 : \theta \leq c, \,\,&\text{i.e., the area under the ROC curve} \notag \\
& \text{is less than a specific value of interest $c$.} \label{eq:nullhypo}
\end{align}

Since we do not know the underlying distribution of the observed information rates we employ a non-parametric approach for significance testing by using the method of bootstrapping \cite{efron1981nonparametric, efron1994introduction,cortes2005confidence, liu2005testing}. Further, since the observed information rates in the vectors $\mathbf{I}$, $\mathbf{DI}_1$, $\mathbf{DI}_2$, and $\mathbf{TE}^\ast$ are correlated in time with high probability, we use a modified block bootstrapping approach \cite{hall1995blocking} to preserve this temporal correlation.
Given a class $\mathcal{P}$ with a total of $N_p$ samples, the block bootstrapping procedure divides this class into overlapping blocks of $l$ samples each, to create a total of $N_p-l+1$ blocks. Bootstrapping is then performed by drawing $m=N_p/l$ blocks with replacement and concatenating them to form a resampled class. We pick $l = N_p^{1/3} \simeq 8$ \cite{hall1995blocking}.   
The significant test is then carried out by performing the following steps:
\begin{enumerate}[ leftmargin = !, align=left, label=\roman*) ]
	\item We perform block bootstrapping on the positive class $\mathcal{P}$ to create a resampled positive class $\hat{\mathcal{P}}_b$, with $b=1,2,\dots,B$. Likewise $\mathcal{N}$ is block bootstrapped to create resampled negative class $\hat{\mathcal{N}}_b$. Here we use $B=2000$ resamples. Bootstrapping is performed independently for the positive and negative classes by strictly resampling from only within that particular class.
	\item  We calculate the area under the ROC curve for the pair of bootstrap resampled classes $\hat{\mathcal{P}}_b$ and $\hat{\mathcal{N}}_b$, and denote it as $\hat{\theta}_b, \forall \, b=1,2,\dots,B$.
	\item The standard deviation of the area under the ROC curve can be calculated as
		\begin{align}
			\sigma_{\hat{\theta}} = \sqrt{\frac{\sum_{b=1}^B (\hat{\theta}_b - \bar{\hat{\theta}})^2}{B-1}},
		\end{align}
	where $\bar{\hat{\theta}}=\sum{\hat{\theta}_b}/B$ is the estimated mean of the area under $B$ resampled ROC curves. 
	\item  Let $z_{\hat{\theta}} = (\hat{\theta} - \bar{\hat{\theta}})/\sigma_{\hat{\theta}}$ denote the normalized random variable corresponding to the bootstrap calculated test statistic $\hat{\theta}$. Then $\lim\limits_{B\to \infty} p_{z_{\hat{\theta}}} \sim \mathbbm{N}(0,1)$, where $p_{z_{\hat{\theta}}}$ is the the empirical probability distribution of $z_{\hat{\theta}} $, and $\mathbbm{N}(0,1)$ is the standard Gaussian distribution with zero mean and unit variance.
	\item We then calculate the probability of observing a value for the area under the ROC curve under the null hypothesis, or in other words, the probability of $\Pr(\hat{\theta}\leq c)$. This one-sided tail probability is known as the p-value of the test and can be easily calculated using the cumulative distribution function (cdf) of the standard normal distribution, ${\Phi(x)\,=\,\frac{1}{\sqrt{2\pi}}\int_{-\infty}^x\exp{\left(\frac{-u^2}{2}\right)}\,du}$, as ${p_c~\triangleq~\Pr(\hat{\theta} \leq c) = \Phi\left( \frac{ c - \bar{\hat{\theta}} }{ \sigma_{\hat{\theta}} } \right)}$. Note that a smaller p-value is evidence against $H_0$, see \eqref{eq:nullhypo}.
	\item  If the p-value is smaller than a significance threshold $\alpha$, we may safely reject the null hypothesis. Therefore the null hypothesis $H_0:\theta \leq c$ will be rejected if $p_c \leq \alpha.$
	For our purpose, we set the significance threshold as $\alpha=0.05$.  
\end{enumerate}

\subsubsection{\comnew{Inferring significantly changing connections}}
\comnew{A connection between an ROI pair is deemed to \textit{change significantly} if it shows a pronounced difference in the information transfer rates between high and degraded quality, thereby indicating that these ROIs show the greatest change in brain (electrical) activity in response to the change of audio quality.}\footnote{\comnew{We emphasize here that the statistical significant results are not used to identify connections with the largest information flow but rather those with the \textit{largest change} in information flow, where the former is neither checked for, nor required to distinguish between the audio qualities. Note that due to the high dimensionality of the measurement space, assessing statistical significance becomes much more difficult if $X^N$, $Y^N$, and $Z^N$ are not jointly Gaussian distributed.}}
Since a larger area under the ROC curve indicates a greater statistical difference between the information rates, significant connections are identified by testing the null hypothesis $H_0:\theta\leq c$ for a sufficiently large value of $c$. However, the optimal choice of the cutoff threshold for inferring the functional connectivity network in the brain is a rather difficult task \cite{van2010comparing, langer2013problem}, and often depends on several criteria including the number of nodes and connections (edges) in the network, the true topology of the underlying functional network, and the nature of the neurophysiological task involved.

For the sake of exposition we choose here $c=0.85$ and test against the null hypothesis $H_0: \theta \leq 0.85$. We perform the significance test via bootstrapping as outlined above and calculate the corresponding p-value $p_{0.85}$ for each information measure and all ROI pairs.
If $p_{0.85} \leq 0.05$, then $H_0$ is rejected, and the ROI pair is \comnew{declared to be a significantly changing connection.}
Fig.~\ref{fig:p75} shows a matrix representation of the net result for all possible combination of ROI pairs. The rows of the matrix represent the source ROI and the columns represent the destination ROI. We observe from these results that CBI is most successful at rejecting the null hypothesis $\theta \leq 0.85$ across all ROI pairs showing the highest discriminability between the high and degraded quality. We also observe that while Massey's directed information and sum transfer entropy perform much better than Kamitake's directed information, they still perform poorly when compared to CBI. As CBI is a bidirectional symmetric measure the inferred change in connectivity between the ROI pairs is also symmetric.
\comnew{Note that there also appears to be a high degree of bidirectional symmetry in the inferred connectivity change using Massey's and Kamitake's directed information and sum transfer entropy.}
Finally, we observe from Fig.~\ref{fig:p75} that every connection identified as significant by Massey's directed information and/or sum transfer entropy is likewise identified as significant by CBI, which is to be expected since CBI is in effect an extension of the two as shown in Proposition~\ref{prop:CBI}.

We observe from the inferred \comnew{connectivity change} matrix for CBI that the most active regions are the temporal lobes (ROIs~5 and 7) along with the frontal lobe (ROIs~1, 2 and 3). This is consistent with the literature that the ventral auditory pathway, which in part includes the auditory cortex and the prefrontal cortex, has a role in auditory-object processing and perception \cite{rauschecker2009maps, wang2008dual, bizley2013and}. There is also significant information flow from both the left and right temporal lobes to the parietal lobe (ROI~6). Interestingly, there appears to be no dorsal information flow from the parietal lobe (ROI 6) to the frontal lobes. Also, no significant connections are detected in the occipital lobe (ROI 8), which houses the primary visual cortex and would not be expected to show significant activity for an auditory-discrimination task.  

\begin{figure}[tbp]
	\vspace{-1ex}
	\centering
	\includegraphics*[width=\columnwidth, trim= 79.5 2 48 .5, clip=true]{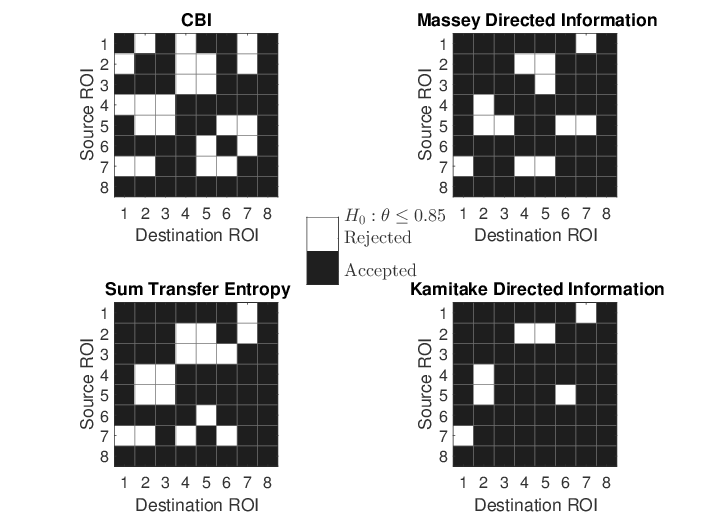}
	\vspace{0.3ex}
	\caption{\small{Inferred \comnew{connectivity change} matrix in response to the change in the perceived audio quality. The null hypothesis tested is $H_0: \theta \leq 0.85$, for which the area under the ROC curve is not a significant connection. The white squares indicate significant connections which show the largest change in their information transfer rates in response to changing audio quality.}
	} \label{fig:p75}
	\vspace{-3.5ex}
\end{figure}

\vspace{-2ex}
\section{Conclusion}
\vspace{-1ex}
We presented a novel information theoretic framework to assess changes in perceived audio quality by directly measuring the EEG response of human subjects listening to time-varying distorted audio. Causal and directional information measures were used to infer the \comnew{change in} connectivity between EEG sensors grouped into ROIs over the cortex. In particular, Massey's directed information, Kamitake's directed information, and sum transfer entropy were each used to measure information flow between ROI pairs while successfully accounting for the influence from all other interacting ROIs using causal conditioning. We also proposed a new information measure which is shown to be a causal bidirectional modification of directed information applied to a generalized cortical network setting, and whose derivation is strongly related to the classical MAC with feedback. Further, we showed that CBI performs significantly better in being able to distinguish between the audio qualities when compared to the other directed information measures. The connectivity results demonstrate that a \comnew{change in information flow} between different brain regions typically occurs as the subjects listen to different audio qualities, with an overall increase in the information transfer when the subjects listen to degraded quality as opposed to high quality audio. We also observe \comnew{significant connections with respect to a change in audio quality} between the temporal and frontal lobes, which is consistent with the regions that would be expected to be actively involved during auditory signal processing in the brain.  

\vspace{-2ex}
\bibliographystyle{IEEEtran} \bibliography{refs}
\vspace {-9ex}
\begin{IEEEbiography}[{\includegraphics[width=1in,height=1.25in ,clip,keepaspectratio]{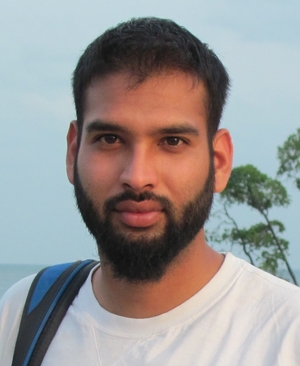}}]{Ketan Mehta}
	received his M.S. (2010) in electrical engineering and Ph.D. (2017), both from New Mexico State University, Las Cruces, USA. From 2012 to 2017 he was a research assistant at the New Mexico State University. His research interests span information theory, signal processing and statistical algorithms with interdisciplinary applications in neural signal processing and cognitive neuroscience.    
\end{IEEEbiography}
\vspace {-8ex}
\begin{IEEEbiography}[{\includegraphics[width=1in,height=1.25in, clip,keepaspectratio]{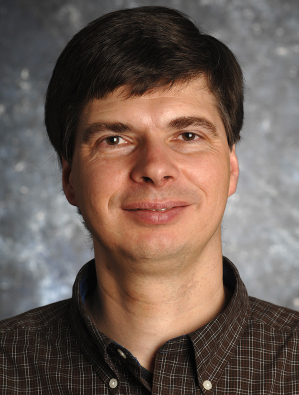}}]{J{\"o}rg Kliewer}
	(S'97--M'99--SM'04) received the
	Dipl.-Ing.~(M.Sc.) degree in electrical engineering from Hamburg
	University of Technology, Hamburg, Germany, in 1993 and the
	Dr.-Ing.~degree (Ph.D.) in electrical engineering from the University
	of Kiel, Germany, in 1999, respectively.
	From 1993 to 1998, he was a research assistant at the University of
	Kiel, and from 1999 to 2004, he was a senior researcher and lecturer
	with the same institution. In 2004, he visited the University of
	Southampton, U.K., for one year, and from 2005 until
	2007, he was with the University of Notre Dame, IN, as a
	Visiting assistant professor. From 2007 until 2013 he was with New Mexico
	State University, Las Cruces, NM, most recently as an associate professor.
	He is now with the New Jersey Institute of Technology, Newark, NJ,
	as an associate professor. His research interests span information and coding theory, graphical models, and statistical algorithms, which includes applications to networked communication and security, data storage, and biology.
	Dr.~Kliewer was the recipient of a Leverhulme Trust Award and a German
	Research Foundation Fellowship Award in 2003 and 2004, respectively.
	He was an Associate Editor of the IEEE
	Transactions on Communications  from 2008 until 2014, and since 2015 serves as an Area Editor for the same journal. He is also an Associate Editor of the IEEE Transactions on Information Theory since 2017 and a member of the editorial board of the IEEE Information Theory Newsletter since 2012.
\end{IEEEbiography}

\end{document}